\documentclass[twocolumn,10pt]{IEEEtran}

\input{def.tex}
\usepackage{dsfont}
\usepackage{minibox}
\DeclareSymbolFont{matha}{OML}{txmi}{m}{it}
\DeclareMathSymbol{\varv}{\mathord}{matha}{118}
\usepackage{eurosym}
\usepackage{hhline}
\usepackage{multicol,bbm}

\makeatletter

\IEEEoverridecommandlockouts
\begin{document}
	\title{Performance Analysis of Cell-Free Massive MIMO Systems: A Stochastic Geometry Approach}
	\author{Anastasios Papazafeiropoulos, Pandelis Kourtessis, Marco Di Renzo, Symeon Chatzinotas, and John M. Senior \vspace{2mm} \\
		\thanks{A. Papazafeiropoulos is with the Communications and Intelligent Systems Research Group, University of Hertfordshire, Hatfield, U. K. and with SnT at the University of Luxembourg, Luxembourg. P. Kourtessis and J.M. Senior are with the Communications and Intelligent Systems Research Group, University of Hertfordshire, Hatfield, U. K. M. Di Renzo is with the Laboratoire des Signaux et Syst\`emes, CNRS, CentraleSup\'elec, Universit\'e Paris Sud, Universit\'e Paris-Saclay, France. S. Chatzinotas is with the SnT at the University of Luxembourg, Luxembourg. E-mails: tapapazaf@gmail.com, \{p.kourtessis,j.m.senior\}@herts.ac.uk, marco.direnzo@l2s.centralesupelec.fr, symeon.chatzinotas@uni.lu.}
			\thanks{This work is supported in part by the National Research Fund, Luxembourg, under the projects ECLECTIC, DISBUS, 5G-SKY.}
	}
	
	\maketitle
	%
	\vspace{-1.4cm}
	
	\begin{abstract}
		Cell-free (CF) massive multiple-input-multiple-output (MIMO) has emerged as an alternative deployment for conventional cellular massive MIMO networks. As revealed by its name, this topology considers no cells, while a large number of multi-antenna access points (APs)  serves simultaneously a smaller number of users over the same time/frequency resources through time-division duplex (TDD) operation. Prior works relied on the strong assumption (quite idealized) that the APs are uniformly distributed, and actually, this randomness was considered during the simulation and not in the analysis. However, in practice, ongoing and future networks become denser and increasingly irregular. Having this in mind, we consider that the AP locations are modeled by means of a Poisson point process (PPP) which is a more realistic model for the spatial randomness than a grid or uniform deployment. In particular, by virtue of stochastic geometry tools, we derive both the downlink coverage probability and achievable rate. Notably, this is the only work providing the coverage probability and shedding light on this aspect of CF massive MIMO systems. Focusing on the extraction of interesting insights, we consider small-cells (SCs) as a benchmark for comparison. Among the findings, CF massive MIMO systems achieve both higher coverage and rate with comparison to SCs due to the properties of favorable propagation, channel hardening, and interference suppression. Especially, we showed for  both architectures  that increasing the AP density results in  a higher coverage which saturates after a certain value and increasing the number of users decreases the achievable rate  but CF massive MIMO systems take advantage of the aforementioned properties, and thus, outperform SCs. In general, the performance gap between CF massive MIMO systems and SCs is enhanced by increasing the AP density. Another interesting observation concerns that a higher path-loss exponent decreases the rate while the users closer to the APs affect more the performance in terms of the rate. 
	\end{abstract}
	
	
	\section{Introduction}
	The landscape of wireless communications is undergoing a rapid revolution~\cite{Shafi2017}. From video streaming and social networking to immersive technologies such as augmented and virtual reality (AR/VR) emerging applications are pushing mobile operators to evolve persistently and perpetually, in order to find the optimal cellular network architecture~\cite{Andrews2014,Ericsson2015,Bastug2017}. The forthcoming fifth-generation (5G) networks have {accepted} the adoption of the massive multiple-input-multiple-output (MIMO) technology~\cite{Wang2018}, where each base station (BS) equips a large number of antennas and exploits the spatial multiplexing of many users on the same time-frequency resources, in order to take advantage of the accompanying high spectral and energy efficiency, reliability, and simple signal processing~\cite{Larsson2014}. A collocated antenna array in each cell and multiple geographically distributed antennas arrays represent the two extreme ends of the topology spectrum~\cite{Larsson2014,Ngo2017}.
	
	\subsection{Prior Work}
	Conventional massive MIMO capitalize on the advantageous channel hardening and the favorable propagation phenomena, being consequences of the law of large numbers, e.g., see~\cite{Ngo2013} and references therein. In particular, channel hardening turns the multi-antenna fading channel gain into nearly deterministic~\cite{Ngo2017a} while favourable propagation turns the users’ channel vectors to almost orthogonal~\cite{Ngo2014}. Fundamentally, massive MIMO results as a scalable form of multi-user MIMO with respect to the number of BSs antennas, where the BSs are aware of the channel state information (CSI) and can employ simple linear precoding~\cite{Larsson2014,Hoydis2013}. In this direction, the relevant research has revealed that the CSI acquisition, limited by the channel coherence block, is degraded by means of the pilot contamination~\cite{Marzetta2010}. In addition, it has been shown that the transceiver hardware impairments have less impact on massive MIMO than on contemporary systems with a finite number of antennas~\cite{Bjornson2015,PapazafComLetter2016,Papazafeiropoulos2017,Papazafeiropoulos2017a,Papazafeiropoulos2018}.
	
	On the other hand, a co-processing technology, known as network MIMO~\cite{Shamai2001}, assumes a set of geographically distributed access points (APs) to serve jointly all the users by utilizing only local CSI at each AP, in order to keep its implementation feasible due to substantial backhaul overhead with a comparison to global CSI~\cite{Karakayali2006,Irmer2011}. In fact, the channels between the APs and the users are estimated by means of an uplink training phase and the property of channel reciprocity under a time-division duplex (TDD) design. Recently, the network MIMO notion has sprung again under the name cell-free (CF) massive MIMO, but herein, the number of APs, serving a smaller number of users, grows large~\cite{Ngo2017}. In such architecture, the cell boundaries are abrogated, and contrary to centralized massive MIMO systems serving poorly cell-edge users due to inter-cell interference, the quality of service is improved because of reduced distances between the users and the APs together with rich diversity against large-scale fading. Therefore, CF massive MIMO design enjoys the key principles from both network and massive MIMO, i.e., it reaps the benefits from its distributed nature being an increased macro-diversity gain while the information exchange among the APs is limited which results in better computational efficiency. In parallel, it exploits channel hardening and favorable propagation. Overall, the CF massive MIMO configuration is quite promising for next-generation systems because many users can be served simultaneously with high quality of service due it to high macro-diversity, low path-losses, and increased expected coverage, since the serving antennas are closer to the user.

	 With little available work in the literature at the beginning, the high potentials of CF massive MIMO have attracted a lot of scientific interest with~\cite{Nayebi2017,Interdonato2016,Ngo2018,Ngo2017, Zhang2018,Hu2019,Buzzi2017a,Buzzi2017,Bashar2018,Chen2018,Interdonato2019}. Specifically, in~\cite{Ngo2017}, the outperformance of CF massive MIMO against small cells (SCs) was depicted by taking into account for the effects of power control, pilot contamination, and imperfect CSI, while in~\cite{Nayebi2017}, the authors studied the performance of CF massive MIMO systems with zero-forcing (ZF) precoding, but with no pilot contamination, i.e., the assigned pilots to the users are mutually orthogonal. In~\cite{Interdonato2016}, downlink training was introduced in the design and it was showed that beamformed pilots improve the performance. Furthermore, in~\cite{Ngo2018}, the APs were enriched with multiple antennas that bring an increase to the array and diversity gains, and maximize energy efficiency. Moreover, in~\cite{Zhang2018}, the authors studied the impact of additive hardware impairments in CF massive MIMO systems while, in~\cite{Hu2019}, the achievable rate was derived by using practical low-resolution analog-to-digital converters. Furthermore, a user-centric approach, where each AP serves a selected number of users, has been suggested in\cite{Buzzi2017a,Buzzi2017} to provide larger achievable rates with reduced backhaul overhead. Also, one of the main challenges of distributed antenna systems, being the limited backhaul, was studied in~\cite{Bashar2018}. In~\cite{Chen2018}, the spatial randomness was taken into account to explore the assumptions that should hold for channel hardening and favorable propagation in CF massive MIMO systems. Regarding practical applications, a very promising implementation of CF Massive MIMO systems is through radio stripes as described in~\cite{Interdonato2019}.

	Unluckily, the majority of works in CF massive MIMO have assumed  a constant number of APs, being uniformly distributed in a finite region but this randomness is taken into account only during the simulations and not in the analysis. However, in practice, the locations of the APs follow highly irregular spatial patterns. 
	In particular, in the case of SCs, tractable and accurate models, describing the realistic randomness of the BS locations for single-input-single-output (SISO) channels were introduced in~\cite{Andrews2011}. Specifically, the authors relied on the theory of Poisson point processes (PPPs) which draws the number of BSs from a Poisson random number generator~\cite{Baccelli2010}, and provided analytical expressions for the coverage probability and the achievable rate. Towards this end, many studies, for example,~\cite{Kountouris2012} and~\cite{Dhillon2013} extended~\cite{Andrews2011} to a multi-user network and multiple antennas, respectively. Actually, this approach, applied in heterogeneous networks (HetNets), facilitates the modeling of the cell densification and provides a more accurate description with  comparison to the grid model for example. In fact, heterogeneous SCs have been suggested as a major technology to be implemented in 5G networks, since it improves dramatically the coverage, the spatial reuse, and boosts the spectral efficiency per unit area (see~\cite{Andrews2012,Kamel2016}, and references therein). 
	\subsection{Motivation}
	This work relies on the paramount observation that the advantageous CF massive MIMO systems have been evaluated under the unrealistic assumption of the uniform placement of the APs in a finite area. Moreover, there is no work in the area of CF massive MIMO systems with PPP distributed APs providing their coverage probability in analytical form. In fact, there are no available results considering that the APs are PPP distributed except the useful study in~\cite{Chen2018} that explored the validity of the favorable propagation and channel hardening properties but did not focus on the performance analysis. In previous works, the APs were assumed to be distributed uniformly that results in pessimistic conclusions. Notably, these models are highly idealized and inaccurate, especially, in the case of CF massive MIMO which are designed based on heterogeneous and ad hoc deployments. Hence, the fundamental question, addressed by this work, is ``how a more realistic randomness regarding the AP locations affects the performance of CF massive MIMO systems?''. Motivated by this arising need, we are going to establish the theoretical framework incorporating the randomness of the AP locations and identify the realistic potentials of CF massive MIMO systems before their final implementation. 
	\subsection{Contributions and Outcomes}
	The main contributions are summarized as follows.
	\begin{itemize}
		\item Contrary to existing work~\cite{Ngo2017}, which has proposed the concept of CF massive MIMO in the case of just uniform placement of the APs, we introduce a more realistic spatial randomness where the APs are PPP distributed. Also, contrary to~\cite{Chen2018} which examined whether the phenomena of channel hardening and favorable propagation appear in CF massive MIMO systems, we focus on their performance analysis in terms of derivation of analytical results.
		
		
		\item We carry out an asymptotic performance analysis by deriving the deterministic equivalent (DE)\footnote{The DE analysis is a tool of random matrix theory (RMT) to achieve deterministic expressions concerning matrices when their size grows large but with a given ratio. The DEs have been widely employed in massive MIMO systems by providing deterministic expressions, and thus, avoiding the need for lengthy Monte Carlo simulations~\cite{Couillet2011,Papazafeiropoulos2015a}.} downlink signal-to-interference-plus-noise ratio (SINR) with maximum ratio transmission (MRT), while we have assumed that the channel is estimated during an uplink training phase including pilot contamination. It is worthwhile to mention that there are no other prior works providing DE expressions for CF massive MIMO systems.
		\item We derive the downlink coverage probability and the achievable rate per user. 
		As far as the authors are aware, these are the only analytical results in the literature concerning the coverage and the rate that account for APs PPP distributed in CF massive MIMO systems. 
		Especially, it is the first work in the literature providing the coverage probability in CF massive MIMO systems. 
		For the sake of comparison, we also present the results corresponding to SCs, where each user is associated with its nearest AP.
		\begin{itemize}
			\item 	We shed light on the coverage probability. 
				We observe that CF massive MIMO systems outperform the comparable SCs architecture. Specifically, the coverage decreases with the target SINR due to increasing interference. Also, as the density of APs increases, the coverage probability increases and saturates at large density.
			Also, a higher threshold limits the coverage probability to lower values. Notably, in all studied conditions, CF massive MIMO systems outperform SCs mostly due to favorable propagation, channel hardening, and interference suppression.
			\item We focus on the downlink achievable user rate and delve into the effects of the duration of the training phase and the number of users. A reduction of the training duration or an increase in the number of users worsens the system performance because of more severe pilot contamination and inter-user interference. Furthermore, the interference in SCs is more dominant than in CF massive MIMO systems because the latter exploit the effect of favorable propagation. Also, a higher average number of APs brings higher diversity gain, and as a result, it increases the rate. The properties of favorable propagation, channel hardening, and interference suppression become more pronounced as the average number (density) of APs increases. Moreover, by increasing the path-loss exponent, it is depicted that the user rate decreases. Based on this observation, we can extract the conclusion that distant users hardly affect the overall rate.
		\end{itemize}
		
	\end{itemize}

	\subsection{Paper Outline} 
	The remainder of this paper is structured as follows. Section~\ref{System} presents the general framework, where the APs in a CF massive MIMO system are PPP distributed. In Section~\ref{estimation}, the channel estimation phase is presented. Next, Section~\ref{downlink} exposes the downlink transmission and the derivation of the SINR when the APs are randomly located. In Sections~\ref{CoverageProbability} and~\ref{AchievableSpectralEfficiency}, we obtain the coverage probability and the achievable user rate are provided by accounting for the spatial randomness. The numerical results are placed in Section~\ref{Numerical}, while Section~\ref{Conclusion} summarises the paper.
	\subsection{Notation}Vectors and matrices are denoted by boldface lower and upper case symbols. The notations $(\cdot)^\T$, $(\cdot)^\H$, and $\tr\!\left( {\cdot} \right)$ express the transpose, Hermitian transpose, and trace operators, respectively. The expectation and variance operators are denoted by $\EE\left[\cdot\right]$ and $\var\left[\cdot\right]$, respectively. The notations $\mathcal{C}^{M \times 1}$ and $\mathcal{C}^{M\times N}$ refer to complex $M$-dimensional vectors and $M\times N$ matrices, respectively. Finally, $\bb \sim \cC\cN{(\b0,\mathbf{\Sigma})}$ represents a circularly symmetric complex Gaussian vector with zero-mean and covariance matrix $\mathbf{\Sigma}$.
	
	\section{System Model}\label{System}
	\subsection{System Core and APs Arrangement}\label{ChannelModel} 
	We focus on a software-defined CF massive MIMO network, where the central processing unit (CPU), being a software-defined network (SDN) controller, manages separately the control and data planes by means of a perfect backhaul. In practice, the backhaul is not perfect, but it is subject to significant limitations~\cite{Marsch2011}. In other words, the SDN controller regulates the operations among the APs cooperating phase-coherently\footnote{Obviously, the introduction of SDN in this scenario seems to be an attractive almost mandatory solution because the coordination of the APs is a complex task missing from simpler systems such as an SCs architecture, where only the data and power control coefficients are the burden of the CPU~\cite{Ngo2017}. Otherwise, the APs coordination will be a difficult task. This argument becomes more meaningful in the CF massive MIMO setting because we consider that the average number of APs is very large.}. Specifically, we assume that a large number of APs each equipped with $N\ge 1$ antennas under a network MIMO concept serves jointly a set of $K$ single-antenna users in the same time-frequency resources\footnote{The adaptation of a scheduling algorithm allows the selection of $K$ from a large set of users. Note that the locations of the users are distributed according to some independent stationary point process~\cite{Kountouris2012}.}. In other words, the network is not partitioned into cells, and each user is served by all APs simultaneously. Actually, the AP locations $\{x_{i}\} \subset \mathbb{R}^{2}$ are generated randomly and follow a two-dimensional homogeneous PPP $\Phi_{\mathrm{AP}}$ with density $\lambda_{\mathrm{AP}}$ $\left[\mathrm{APs}/\mathrm{km}^{2} \right] $. 
	Moreover, let a finite-sized geographic area (Euclidean plane) $\mathcal{A}$ occupying space $S\!\left( \mathcal{A} \right)$ $\mathrm{m}^{2}$, which denotes the Lebesgue measure of the set $\mathcal{A}$. 
	In a specific realization of the PPP $\Phi_{\mathrm{AP}}$, the number of APs $M$ is a random variable obeying to the Poisson distribution with mean $ \tilde{M}=\EE\left[ M\right]$ given by
	\begin{align}
	\tilde{M} =\lambda_{\mathrm{AP}} S\!\left( \mathcal{A} \right).\label{meanValue} 
	\end{align}
	
	In this regard, in a specific network realization, the total number of antennas in $\mathcal{A}$, denoted by $\mathcal{W}=MN$, is a Poisson random variable with mean $\EE\left[ \mathcal{W}\right]= \tilde{M} N$. 
	Also, in most realizations, we assume  $ \mathcal{W}\gg K $ corresponding to a CF massive MIMO scenario. Under this condition, it is more possible that distant users
	can enjoy coverage by a close AP similar to the provided coverage to more central users. Moreover, thanks to Slivnyak's theorem, it is sufficient to focus on a typical user, in order to conduct the analysis and investigate the performance of the network~\cite{Chiu2013a}. Note that the typical user corresponds to a user chosen at random from amongst all users in the network \footnote{ The average network performance, met by randomly located users in the network, is equivalent to the spatially averaged network statistics followed by the typical user.}. Without loss
	of generality and for the ease of exposition, we assume that the typical user is located at the origin.

For the sake of convenience, Table~\ref{ParameterValues2} summarizes the notation used throughout the paper.
	\begin{table*}
		\caption{Parameters Values for Numerical Results }
		\begin{center}
			\begin{tabulary}{\columnwidth}{ | c | c | }\hline
					{\bf Notation} &{\bf Description}\\ \hline
					$ N $, $ M $, and $ \mathcal{W} $& Number of Antennas/AP, number of APs in a PPP realization, and total number of antennas \\ \hline
					$ K $ & Number of users\\ \hline
					$\bPhi_{\mathrm{AP}}$ and $\lambda_{\mathrm{AP}}$& Poisson point process of APs and its density\\ \hline
					$S\!\left( \mathcal{A} \right)$ & Space of an area $ \mathcal{A} $\\ \hline
					$ \tilde{M} $ and $ \mathcal{\tilde{W}} $ & Mean number of APs and antennas \\ \hline
					$ B_{\mathrm{c}}$ and $ T_{\mathrm{c}} $ & Coherence bandwidth and time \\ \hline
					$\tau_{\mathrm{tr}}$, $\tau_{\mathrm{u}}$, and $ \tau_{\mathrm{d}} $ & Duration of uplink training, uplink and dowlink transmission phases \\ \hline
					$\bh_{mk}$, $l_{mk}$, and $\bg_{mk}$& Channel, path loss, and small-scale fading vectors between $ m $th AP and user $ k $\\ \hline
					$r_{mk}$, $ \al $& Distance between $ m $th AP and user $ k $, path loss exponent\\ \hline
					$\bpsi_{k}$&Normalized pilot sequence\\ \hline
					${\rho}_{\mathrm{tr}}$ and $ {{\rho}_{\mathrm{d}}} $	& 		Uplink training transmit power per pilot symbol and downlink transmit power\\ \hline
					$\hat{\bh}_{mk}$ and $\tilde{\bee}_{mk}$ & Estimated and error channel vectors \\ \hline
					$\bL_{k}$ and		$\bPhi_{k}$ & Covariances of ${\bh}_{k}$ and $\hat{\bh}_{k}$ \\ \hline
					$\mu$ & Normalization precoding parameter \\ \hline
					$ \gamma_{k}$ and $ \bar{\gamma}_{k} $ & Statistical and DE SINR of user $ k $ \\ \hline
					$ P_{\mathrm{c}}^{\mathrm{cf}} $ & Coverage probability \\ \hline
					$ {R}_{k}^{\mathrm{cf}} $ and $ \check{R}_{k}^{\mathrm{cf}} $ & Achievable rate and its lower bound \\ \hline
				\end{tabulary}\label{ParameterValues2}
		\end{center}
	\end{table*}
	\subsection{Channel Model}\label{ChannelModel} 
	In our analysis, we consider both small-scale fading and independent large-scale fading in terms of path-loss. The independence relies on the fact that the latter stays static for several coherence intervals, while the former changes faster contingent on the user mobility, i.e., it is assumed static for one coherence interval, but it changes from one interval to the next. As a typical example, the large-scale fading should stay
	constant for a duration of at least 40 coherence intervals~\cite{Rappaport1996}. Note that each coherence interval with coherence time $\tau_\mathrm{c}=B_{\mathrm{c}}T_{\mathrm{c}}$ samples (channel uses) incorporates three phases, where $B_{\mathrm{c}}$ in $\mathrm{Hz}$ and $T_{\mathrm{c}}$ in $\mathrm{s}$ denote the coherence bandwidth and time, respectively. Specifically, we include the uplink training phase of $\tau_{\mathrm{tr}}$ symbols as well as the uplink and downlink data transmission phases of $\tau_{\mathrm{u}}$ and $\tau_{\mathrm{d}}$ samples, respectively\footnote{In conventional massive MIMO systems and in CF massive MIMO systems, a downlink training phase does not take place because the users take into advantage of the channel hardening and need only the average effective channel gain instead of the actual effective gain~\cite{Ngo2017}. However,~\cite{Chen2018} showed that CF massive MIMO systems do not always experience channel hardening except certain conditions such as small path-loss exponent and relatively large distance among users. Herein, we assume that the required conditions for channel hardening and favorable propagation are met.}. The two data transmission phases assume identical channels based on the property of channel reciprocity being achievable under TDD operation and calibration of the hardware chains. In this work, we focus on the uplink training and downlink data transmission phases.

	We consider a specific realization of the PPP $\Phi_{\mathrm{AP}}$, where the number of the APs is $M$. Let $\bh_{mk}$ be the $N\times 1$ channel vector between the $m$th AP and the typical user denoted henceforth by the arbitrary index $k$. In particular, the channel vector $\bh_{mk}$ is expressed as
	\begin{align}
	\bh_{mk}= l_{mk}^{1/2}\bg_{mk},
	\end{align}
	where $l_{mk}$ and $\bg_{mk}$ with $m=1,\ldots,M$ and $k=1,\ldots,K$ represent the independent large-scale and small-scale fadings between the $m$th AP and the typical user. 
	Specifically, the large-scale fading considers geometric attenuation (path-loss) by means of  $l_{mk}\left( r_{mk} \right)=\min\left( 1, r_{mk}^{-\al}\right)$ being a non-singular bounded 
		pathloss model with $\al>0$ being the path-loss exponent while $r_{mk}$ expresses the distance between the $m$th AP and the typical user~\cite{haenggi2009interference}. 
		Note that an unbounded path-loss model such as $l_{mk}\left( r_{mk} \right)= r_{mk}^{-\al}$ is not appropriate in the case of CF massive MIMO systems, where an AP can approach arbitrarily close to a user, resulting in unrealistically high power gain~\cite{Chen2018}.
	Especially, regarding the distance $r_{mk}$ from the serving APs to the typical UE, which actually involves the communication between a random AP and a random user, we assume that it follows the uniform distribution in $\mathcal{A}$. Similarly, the distances from other users are independent and follow the uniform distribution. Furthermore, $\bg_{mk}$, modeling Rayleigh fading, consists of small-scale fading elements, which are assumed to be independent and identically distributed (i.i.d.) $\mathcal{CN}\left( 0,1 \right)$ random variables since, in practice, the groups of scatterers between each AP and each user, distributed in a wide area, are different. Given that both line and non-line of sight signals may appear in CF massive MIMO systems, the application of other fading models could be considered in future works with techniques found in~\cite{ ElSawy2013}.

	\section{Uplink Channel Estimation}\label{estimation}
	Given that the promised multiplexing gains of broadcast channels demand the knowledge of CSIT, an uplink training phase is necessary to allow the APs to compute the estimates $\hat{\bg}_{mk}$ of their local channels. Nevertheless, the re-use of pilot sequences emerges an effect known as pilot contamination, which is more prominent for massive MIMO than in conventional MIMO systems~\cite{Marzetta2010}.
	
	For this reason, in each realization of the network, there is an uplink training phase, where all $K$ users send simultaneously non-orthogonal pilot sequences with duration equal to $\tau_{\mathrm{\tr}}<K $ samples due to the limited length of the coherence interval. Note that the subscript $\mathrm{tr}$ denotes the training stage.
	By denoting $\bpsi_{k}\in \mathbb{C}^{\tau_{ \mathrm{tr}} \times 1} $ the normalized sequence of the $k$th user with $\|\bpsi_{k}\|^{2}=1$, the $N \times \tau_{ \mathrm{tr} } $  received  channel by the $m$th AP is given by
	\begin{align}
	\!\!\!\tilde{\by}_{{\mathrm{tr}},m}&
	\!= \! \sum_{i=1}^{K}\sqrt{\tau_{\mathrm{tr}} \rho_{\mathrm{tr}}}l_{mi}^{1/2} \bg_{mi}\bpsi_{i}^{\H}\!+\!\bn_{{\mathrm{tr}},m},\label{eq:Ypt2}
	\end{align}
	where $\bn_{{\mathrm{tr}},m}$ is the $N \times \mathrm{tr}$ additive noise vector at the $m$th AP consisting of i.i.d. $\mathcal{CN}\left( 0,1 \right)$ random elements, and $\rho_{\mathrm{tr}}$ is the normalized signal-to-noise ratio (SNR). By assuming orthogonality among the pilot sequences, the $m$th AP estimates the channel by projecting $\tilde{\by}_{{\mathrm{tr}},mk}$ onto $\frac{1}{\sqrt{\tau_{\mathrm{tr}} \rho_{\mathrm{tr}}}}\bpsi_{k}$, i.e., we have
	\begin{align}
	\tilde{\by}_{mk}&=\frac{1}{\sqrt{\tau_{\mathrm{tr}} \rho_{\mathrm{tr}}}}\tilde{\by}_{\mathrm{tr},m}\bpsi_{k}\label{eq:Ypt29}\\
	&= \bg_{mk}l_{mk}^{1/2}\!+\! \sum_{i\ne k}^{K}l_{mi}^{1/2} \bg_{mi}\bpsi_{i}^{\H}\bpsi_{k}\!+\!\frac{1}{\sqrt{\tau_{\mathrm{tr}} \rho_{\mathrm{tr}}}}\bn_{\mathrm{tr},m}\bpsi_{k},\label{eq:Ypt3}
	\end{align}
	where the summation in the second term corresponds to the multi-user interference. Actually, this term is the source of pilot contamination.
	Assuming, that the distance $r_{mk}$ is known a priori, the $m$th AP obtains the linear minimum mean-squared error (MMSE) estimate according to \cite{Verdu1998} as
	\begin{align}
	\hat{\bh}_{mk}\!&=\!{\mathrm{E}\!\left[\bh_{mk}^{\H}\tilde{\by}_{\mathrm{tr},mk} \right]}\left(\mathrm{E}\!\left[\tilde{\by}_{\mathrm{tr},mk}\tilde{\by}_{\mathrm{tr},mk}^{\H}\right]\right)^{-1}\tilde{\by}_{mk}\nn\\
	&=\frac{{l_{mk}}}{\sum_{i=1}^{K}|\bpsi_{i}^{\H}\bpsi_{k}|^{2}l_{mi}+\frac{1}{{{{\tau_{\mathrm{tr}} \rho_{\mathrm{tr}}}}}}}\tilde{\by}_{mk}.\label{estimatedChannel1} 
	\end{align}
	

	Having obtained the estimated channel vector $\hat{\bh}_{mk}$, the estimation error vector, based on the orthogonality property of MMSE estimation, is written $\tilde{\bee}_{mk}={\bh}_{mk}-\hat{\bh}_{mk}$. The estimated channel and estimation error vectors are uncorrelated and Gaussian distributed with $N$ identical elements having zero mean and variances given by
	\begin{align}
	\sigma_{mk}^{2}=\frac{l_{mk}^{2}}{ d_{m}}
	\end{align}
	and 
	\begin{align}
	\tilde{\sigma}_{mk}^{2}\!=\!l_{mk}\left( 1-\frac{l_{mk}}{ d_{m}} \right),
	\end{align}
	where $d_{m}=\left( \sum_{i=1}^{K}|\bpsi_{i}^{\H}\bpsi_{k}|^{2}l_{mi}+\frac{1}{{\tau_{\mathrm{tr}} \rho_{\mathrm{tr}}}} \right)$. Hence, we have $\bh_{mk}\in\mathbb{C}^{{N}\times 1}\sim\mathcal{CN}\left( \b0, l_{mk}\Id_{N}\right)$, $\hat{\bh}_{mk}\in\mathbb{C}^{{N}\times 1}\sim\mathcal{CN}\left( \b0, \sigma_{mk}^{2}\Id_{N}\right)$ and $\tilde{\bee}_{k}\in\mathbb{C}^{{N}\times 1}\sim\mathcal{CN}\left( \b0,\tilde{\sigma}_{mk}^{2}\Id_{N}\right)$. 
	At this point, it is better for the sake of following algebraic manipulations to denote the vectors $\bh_{k}= [\bh_{1k}^T \cdots \bh_{Mk}^T]^{\T}\in\mathbb{C}^{\mathcal{W}\times 1}\sim\mathcal{CN}\left( \b0, \bL_{k}\right)$, $\hat{\bh}_{k}= [\hat{\bh}_{1k}^T \cdots \hat{\bh}_{Mk}^T]^{\T}\in\mathbb{C}^{\mathcal{W}\times 1}\sim\mathcal{CN}\left( \b0, \bPhi_{k}\right)$ and $\tilde{\bee}_{k}\in\mathbb{C}^{\mathcal{W}\times 1}\sim\mathcal{CN}\left( \b0,\bL_{k}-\bPhi_{k}\right)$, where the matrices $\bL_{k}$, $\bPhi_{k}= \bL_{k}^{2}\bD^{-1}$, and $\bD$ are $\mathcal{W}\times \mathcal{W}$ are block diagonal, i.e., $\bL_{k}=\mathrm{diag}\left(l_{1k}\Id_{N}, \ldots, l_{Mk}\Id_{N}\right)$, $ \bPhi_{k}=\mathrm{diag}\left(\sigma_{1k}^{2}\Id_{N},\ldots, \sigma_{Mk}^{2}\Id_{N}\right)$, and $\bD=\mathrm{diag}\left(d_{1}\Id_{N}\ldots,d_{M}\Id_{N}\right)$, respectively. In addition, we denote $\bC_{k}=\bPhi_{k}^{-1}$ with $ \bC_{k}=\mathrm{diag}\left(c_{1k}\Id_{N},\ldots,c_{Mk}\Id_{N}\right)$, where $c_{mk}=\sigma_{mk}^{-2}$.

	\section{Downlink Transmission}\label{downlink} 
	This section elaborates on the modeling and characterization of the downlink transmission in one realization of the network, and aims at presenting the downlink SINR, when the APs are PPP distributed and apply conjugate beamforming while the system is impaired by pilot contamination. Having in mind that the users are jointly served by the coordinated APs, we highlight that the received signal by the typical user is given by
	\begin{align}
	y_{\mathrm{d},k}&=\sqrt{\rho_{\mathrm{d}}}\sum_{i\in \Phi_{\mathrm{AP}}} \tilde{\bh}_{i}^{\H} \bs_{i}+z_{\mathrm{d},k}\label{signal} ,
	\end{align}
	where ${\rho_{\mathrm{d}}}$ is the downlink transmit power, $ \tilde{\bh}_{i}$ is the $N \times 1 $ channel vector between the associated AP located at $x_{i}\in \mathbb{R}^{2}$ and the typical user including large and small-scale fadings, $z_{\mathrm{d},k}\sim \mathcal{CN}\left( 0,1 \right)$ is the additive Gaussian noise at the $k$th user, and $\bs_{i}$ denotes the transmitted signal from the $i$th AP. 

	Given that the number of PPP distributed APs in the area $\mathcal{A}$ is $M$, we can rewrite~\eqref{signal} as
	\begin{align}
	y_{\mathrm{d},k}=\sqrt{\rho_{\mathrm{d}}}\sum_{m=1}^{M}\bh_{mk}^{\H} \bs_{m}+z_{\mathrm{d},k},\label{signal1}
	\end{align}
	where $\bh_{mk}$ is the channel between the $m$th AP and user $k$ while $s_{m}$ denotes the transmitted signal from the $m$th associated AP. The transmit signal is written as
	\begin{align}
	\bs_{m}=\sqrt{ \mu } \sum_{k=1}^{K} \bff_{mk} q_{k}\label{ZF}
	\end{align}
	with $q_{k} \in\mathcal{C}$ being the transmit data symbol for the typical user satisfying $\EE\left[ |q_{k}|^{2}\right]=1 $. Actually, the overall transmit signal to users can be written in a vector notation as $\bq = \big[q_{1},\dots,~q_{K}\big]^\T \in \mathbb{C}^{K}\sim \mathcal{CN}(\b0,\bI_{K})$ for all users. Moreover, ${\bff}_{mk}$ represents the $\left( m,k \right)$th element of a linear precoder. In order to avoid sharing channel state information between the APs, we assume scaled conjugate beamforming. We select conjugate beamforming precoding because of its computational efficiency and good performance in both massive MIMO and SCs designs~\cite{Larsson2014,Ngo2017}.  Thus, the expression of the precoder is ${\bff}_{mk}=c_{mk}\hat{\bh}_{mk}$. Regarding the scaling, it relies on a  statistical channel inversion power-control policy 
	that also eases the algebraic manipulations henceforth~\cite{Bjoernson2016}. Also, $\mu$ is a normalization parameter obtained by means of the constraint of the transmit power $\mathbb{E}\left[{ \rho_{\mathrm{d}}}{}\bs\bs^{\H}\right]=\rho_{\mathrm{d}}$. Hence, we have
	\begin{align}
	\mu=\frac{1}{\EE \left[ \tr\bF_{m}\bF_{m}^\H\right]}, \label{eq:lamda} 
	\end{align}
	where $\bF_{m}=\left[\bff_{m1} \cdots \bff_{mK} \right] \in \mathbb{C}^{N \times K}$ is the precoding matrix.

	Taking into account for the imperfect CSIT due to pilot contamination (see~\eqref{eq:Ypt3}), the received signal by the typical user, given by~\eqref{signal1}, is written as
	\begin{align}
	&\!\!\!y_{\mathrm{d},k}=\sqrt{ \mu \rho_{\mathrm{d}} } \sum_{m=1}^{M}\sum_{i=1}^{K}c_{mi}{\bh}_{mk}^{\H} \hat{\bh}_{mi}q_{i}+z_{\mathrm{d},k}\label{filteredsignal10}\\
	&\!\!\!=\sqrt{ \mu \rho_{\mathrm{d}} } \EE\left[ 
	\bh_{k}^{\H}\bC_{k}\hat{\bh}_{k}\right]q_{k}+\sqrt{ \mu \rho_{\mathrm{d}} } {\bh}_{k}^{\H}\bC_{k}\hat{\bh}_{k}q_{k}
	\nn\\
	& \!\!\!-\sqrt{ \mu \rho_{\mathrm{d}} }\EE\left[ {\bh}_{k}^{\H}\bC_{k}\hat{\bh}_{k}\right]q_{k}+\sqrt{ \mu \rho_{\mathrm{d}} }\sum_{i\ne k}^{K}{\bh}_{k}^{\H}\bC_{i}\hat{\bh}_{i}q_{i}+ z_{\mathrm{d},k}, \label{filteredsignal}
	\end{align}
	where the second and fourth terms in \eqref{filteredsignal} describe the desired signal and the multi-user interference. Note that we use similar techniques to~\cite{Medard2000}, i.e., \eqref{filteredsignal10} has been transformed to \eqref{filteredsignal} for the derivation of the SINR provided below since the users are not aware of the instantaneous CSI, but only of its statistics which can be easily acquired, especially, if they change over a long-time scale. Hence, user $k$ has knowledge of only $ \EE\left[ {\bh}_{k}^{\H} \bC_{k}\hat{\bh}_{k}\right]$. In fact, similar to the well-established bounding technique in~\cite{Medard2000}, if we consider that~\eqref{filteredsignal} represents a single-input single-output (SISO) system, the effective SINR of the downlink transmission from all the APs to the typical user under imperfect CSIT, conditioned on the distances of APs $l_{mk}$ for $m=1,\ldots, M$, is given by
	\begin{align}
	\gamma_{k}=\frac{ \Big|\EE\left[ {\bh}_{k}^{\H}\bC_{k}\hat{\bh}_{k}\right]\Big|^{2}}{ \var\left[ {\bh}_{k}^{\H}\bC_{k}\hat{\bh}_{k}\right]+\sum_{i\ne k}^{K} \EE\left[ \Big| {\bh}_{k}^{\H}\bC_{i}\hat{\bh}_{i}\Big|^{2}\right] +\frac{1}{{ \mu \rho_{\mathrm{d}} }}},\label{SINR} 
	\end{align}
	where we assume that the APs treat the unknown terms as uncorrelated additive noise. According to~\cite[Fig. 2]{Ngo2017}, the achievable rate, given by~\eqref{SINR}, provides a rigorous bound close to the achievable rate corresponding to the scenario where the users know the instantaneous channel gain.
	
	As specified by its concept, a CF massive MIMO network comprises a very large number of APs distributed across a geographic area. Hence, relied on the theory of DE analysis which is a common mathematical tool in the large MIMO literature~\cite{Hachem2007,Couillet2011,Papazafeiropoulos2015a}, we can apply it in the proposed CF massive MIMO system, and obtain the asymptotic SINR conditioned on the distances of APs as $K,~\mathcal{W} \rightarrow \infty$, while the finite ratio ${K}/{\mathcal{W}}$ is kept constant. Actually, the definition of DEs, first met in~\cite{Hachem2007} follows.
\begin{definition}[Deterministic Equivalent~\cite{Hachem2007}]
			The deterministic equivalent of a sequence of random complex values $ \left(X_{n}\right)_{n\ge1} $ is a deterministic 	sequence $ \left(\bar{X}_{n}\right)_{n\ge1} $, which approximates $ X_{n} $ such that
			\begin{align}
			X_{n}-\bar{X}_{n} \xrightarrow[ n \rightarrow \infty]{\mbox{a.s.}}0,
			\end{align}
		\end{definition}
		where $ \xrightarrow[ n \rightarrow \infty]{\mbox{a.s.}}0 $ is taken to mean almost sure convergence.
	
	As far as the authors are aware, the DE analysis is applied for the first time in the area of CF massive MIMO. Remarkably, the literature and the simulations in Section~\ref{Numerical} exhibit that the proposed result is of high practical value because of two reasons. First, the result is tight even for conventional system dimensions, i.e., when $20$ APs serve $10$ users. The second reason lies in the fact that a statistical description of the SINR is intractable because of i) the different path-losses from the different APs constituting the desired signal, ii) the 
	cross-products of the path-losses from the different interferers in the denominator can be correlated with the numerator because they contain common path-loss terms.
	
	Conditioned on the distances of APs, the deterministic SINR $\bar{\gamma}_{k}$, obtained such that $\gamma_{k}-\bar{\gamma}_{k}\xrightarrow[ M \rightarrow \infty]{\mbox{a.s.}}0$, is provided below. \begin{proposition}\label{PropDetSINR} 
		Given a realization of $\Phi_{\mathrm{AP}}$ and conditioned on the APs distances, the deterministic SINR of the downlink transmission from the PPP distributed APs to the typical user in a CF massive MIMO system, accounting for pilot contamination and conjugate beamforming, is given by
		\begin{align}
		\bar{\gamma}_{k}\asymp\frac{\mathcal{W}}{ \frac{{1}}{\mathcal{W}} \sum_{i=1}^{K}\tr \bD \bL_{i}^{-2}\left( \bL_{k}+\frac{\mathcal{W}}{\rho_\mathrm{d}} \right)-1}.\label{DESINR} 
		\end{align}
	\end{proposition}
	\begin{proof}
		See Appendix~\ref{SINRproof}.
	\end{proof}
	
	Note that~\eqref{DESINR} holds for any given realization of $\Phi_{\mathrm{AP}}$. In other words, this SINR hides the randomness regarding the AP locations, which is found at the path-losses between the APs and the users. Hence, in order to study the impact of AP density, we have to derive its expectation with respect to the distances.
		Specifcally, $ M $ is found in both $ W $ and inside the trace as one could see in the element-wise expression given by~\ref{gamma1}. By taking the expectation and applying~\cite[Lemma 1]{Zhang2014a}, we have
		\begin{align}
		\bar{\gamma}_{k}\asymp\frac{ \EE\left[M\right] N}{\EE\left[ \frac{{1}}{M} \sum_{i=1}^{K}\sum_{m=1}^{M} d_{m}l_{mi}^{-2}\left( l_{mk}+\frac{MN}{\rho_{\mathrm{d}}}\right)-1
			\right]}.\label{gamma5}
		\end{align}Following a procedure as in Appendix~\ref{CoverageProbabilityproof}, we result in that $ \bar{\gamma}_{k} $ does not depend on the AP density. This property is known as SINR invariance and holds for single-slope path loss models~\cite{Andrews2016}.
	
Regarding the other primary system parameters, $ \bar{\gamma}_{k}\ $ in~\ref{DESINR} saturates with increasing the number of antennas per AP $ N $. Also, when $ \rho_\mathrm{d}\to \infty$, i.e., in the high SNR regime, the SINR reaches a ceiling. Moreover, the SINR decreases with $ K $ and with the severity of pilot contamination.

	\section{Coverage Probability}\label{CoverageProbability}
	The focal point of this section is to shed light on the coverage potentials of CF massive MIMO systems in the realistic setting where the APs are randomly located. Given that the coverage probability of such a system has not been presented before, the first task is to provide a formal definition. The next step is the presentation of the result bringing on the surface its dependence on the system parameters. The derivation, provided in Appendix~\ref{CoverageProbabilityproof}, encompasses techniques and tools from stochastic geometry. Notably, we result in the first expression in the literature that describes the coverage probability of a CF massive MIMO system, being actually the complementary cumulative distribution function (CCDF) of the SINR.
	{\begin{definition}[\!\!\cite{Dhillon2013,Papazafeiropoulos2017}]\label{def1}
			A typical user is in coverage in a CF massive MIMO system if the downlink SINR from the randomly located APs in the network is higher
			than the target SINR $T$.
	\end{definition}}
	
	\begin{theorem}\label{theoremCoverageProbability} 
		The downlink coverage probability of a pilot contaminated CF massive MIMO network, where the APs are PPP distributed and undergo a single-slope path loss model while employing conjugate beamforming, is lower bounded by~\eqref{pc1}, or equivalently~\eqref{pc5} shown at the top of next page, where $\tilde{\mathcal{W}}=\EE\left[ \mathcal{W}\right] $ and $\eta=\tilde{\mathcal{W}} \left( \tilde{\mathcal{W}}! \right)^{-\frac{1}{\tilde{\mathcal{W}}}}$.
		\begin{longequation*}[tp]
			\begin{align}
			P_{\mathrm{c}}^{\mathrm{cf}}&\ge
			\sum^{\tilde{\mathcal{W}}}_{n=1} \!\binom{\tilde{\mathcal{W}}}{n}\!\left( -1 \right)^{n+1} e^{\!\! -{n\eta {T}}\left( {\frac{ K }{\al \pi \rho_{\mathrm{d}}}}\left( \sum_{j=1}^{K} |\bpsi_{j}^{\H}\bpsi_{k}|^{2} \left( \al \rho_{\mathrm{d}}+ \mathcal{\tilde{W}}\left(\al-2\right) \right) +\frac{ \left( \al-2 \right)\rho_{\mathrm{d}} +\mathcal{\tilde{W}}\left(\al-1\right)}{{\tau_{\mathrm{tr}} \rho_{\mathrm{tr}}}} \right)-1\right) \!\!}\label{pc1} \\
			& = 1-\left(1-e^{\!\! -{\eta {T}}\left( {\frac{ K }{\al \pi \rho_{\mathrm{d}}}}\left( \sum_{j=1}^{K} |\bpsi_{j}^{\H}\bpsi_{k}|^{2} \left( \al \rho_{\mathrm{d}}+ \mathcal{\tilde{W}}\left(\al-2\right) \right) +\frac{ \left( \al-2 \right)\rho_{\mathrm{d}} +\mathcal{\tilde{W}}\left(\al-1\right)}{{\tau_{\mathrm{tr}} \rho_{\mathrm{tr}}}} \right)-1\right) \!\!}\right )^{\mathcal{\tilde{W}}}\label{pc5}.
			\end{align}
		\hrulefill
		\end{longequation*}

	\end{theorem}
	\begin{proof}
		See Appendix~\ref{CoverageProbabilityproof}.
	\end{proof}
	
	Focusing on~\eqref{pc5}, we observe better the dependence of the coverage probabilty from the system parameters. In particular, we notice the decrease of $ P_{\mathrm{c}}^{\mathrm{cf}} $ with $ K $ being the number of users. Also, the more severe the pilot contamination is, the lower the coverage probability becomes. A similar behavior results by increasing the target SINR $T$. In fact, if $T \to \infty$, the coverage probability becomes zero. In addition, if the path-loss exponent $\al>2 $ increases, $ P_{\mathrm{c}}^{\mathrm{cf}} $ decreases as expected. Also, in the high SNR regime ($ \rho_{\mathrm{d}} \to \infty $), the coverage probability saturates which means that it is interference limited while when $ \rho_{\mathrm{d}} \to 0 $ it is noise limited since $ P_{\mathrm{c}}^{\mathrm{cf}}\to 0 $. Furthermore, the dependence from the AP density and the number of antennas per AP is given indirectly by means of $ \mathcal{\tilde {W}} $. However, the coverage probability is a complicated function of $ \mathcal{\tilde {W}} $ and the dependence from the corresponding parameter can be shown only by means of numerical results. In Sec.~\ref{Numerical}, it is depicted that $ P_{\mathrm{c}}^{\mathrm{cf}}$ increases with $ \lambda_{\mathrm{AP}} $ and saturates when the AP density becomes large. This saturation is appeared in SCs too~\cite{Andrews2016} in the case of single-slope path loss models. A similar behavior is observed regarding the number of antennas per AP.
	\section{Achievable Rate}\label{AchievableSpectralEfficiency}
	Herein, we provide a closed-form expression of the downlink achievable rate in a CF massive MIMO system. Specifically, the following lemma allows to obtain a tractable lower bound for a large number of APs. 
	
	\begin{lemma}[\cite{Hoydis2013}]
		The downlink ergodic channel capacity of the typical user $k$ in a CF massive MIMO system with conjugate beamforming, PPP distributed AP, and a single-slope path loss model
		is lower bounded by the average achievable rate given by
		\begin{align}
		R_{k}^{\mathrm{cf}}=\left( 1-\frac{ \tau_{\mathrm{\tr}}}{ \tau_{c}} \right)\EE\left[ \log_{2} \left( 1+ \bar{\gamma}_{k}\right)\right] ~~~~~~\mathrm{b/s/Hz},\label{Ratebar} 
		\end{align}
		where $\tau_{c}$ is the channel coherence interval in number of
		samples, $\tau_{\mathrm{\tr}}$ is the duration of the uplink training phase, and $\bar{\gamma}_{k}$ is given by~\eqref{DESINR}. 
	\end{lemma}
	
	Given that the terms in $\bar{\gamma}_{k}$ are actually averaged over the small-scale fading, the expectation in the previous lemma applies to the remaining statistical variables, which are the AP distances. Regarding the pre-log factor, it concerns the pilot overhead. In order to avoid intractable lengthy numerical evaluations of the integrals with respect to the AP distances, we apply Jensen's inequality. The following proposition presents a closed-form expression for the downlink achievable $R_{k}$.
	\begin{theorem}\label{PropDetSINRDistances2} 
		A lower bound of the downlink average achievable rate per user with conjugate beamforming in a CF massive MIMO system with PPP distributed APs is expressed by
		\begin{align}
		\check{R}_{k}^{\mathrm{cf}}=\left( 1-\frac{ \tau_{\mathrm{\tr}}}{ \tau_{c}} \right) \log_{2} \left( 1+ \check{\gamma}_{k}\right) ~~~~~~\mathrm{b/s/Hz},\label{Ratebar1} 
		\end{align}
		where $\check{\gamma}_{k}$ is obtained as shown in~\eqref{DESINR1} at the top of the next page.
		\begin{longequation*}[tp]

			\begin{align}
			\check{\gamma}_{k}= \lambda_{\mathrm{AP}}N \left( {\frac{ K }{\al \pi \rho_{\mathrm{d}}}}\left( \sum_{j=1}^{K} |\bpsi_{j}^{\H}\bpsi_{k}|^{2} \left( \al \rho_{\mathrm{d}}+ N\left (\al-2\right ) \right) +\frac{ \left( \al-2 \right)\rho_{\mathrm{d}} +N\left (\al-1\right )}{{\tau_{\mathrm{tr}} \rho_{\mathrm{tr}}}} \right)-1 \right)^{-1} .\label{DESINR1} 
			\end{align}
			\hrule
		\end{longequation*}
		
	\end{theorem}
	\begin{proof}
		See Appendix~\ref{SINRproofDistances2}.
	\end{proof}
	
	Basically, the impact of the system parameters on the achievable rate is shown by means of $ \check{\gamma}_{k} $. Hence, given that $ \check{\gamma}_{k} $ decreases with the number of users $ K $ and pilot contamination as can be seen by~\eqref{DESINR1}, the corresponding rate decreases as well. Moreover, the achievable rate increases with the number of antennas per AP $ N $, the AP density $ \lambda_{\mathrm{AP}}$, and the transmit power. However, the rate saturates when the number of antennas per AP $ N $ becomes large as expected. In addition, we notice a ceiling at the rate at high $ \rho_{\mathrm{d}} $ but it keeps increasing with the AP density. A similar behavior regarding the AP density is met in small cell systems~\cite{Andrews2016} when a single-slope path loss model is considered. Moreover, the rate decreases with the path-loss exponent $ \al $.

	\section{Numerical Results}\label{Numerical} 
	In this section, we illustrate and discuss the behavior of PPP located APs in a CF architecture for the first time in the corresponding literature since prior works have not taken into account a realistic and well-accepted model for the randomness of APs positions in the analysis. We focus on the analytical expressions concerning the coverage probability $P_{\mathrm{c}}^{\mathrm{cf}}$ and the achievable rate $\check{R}_{k}^{\mathrm{cf}}$, which are provided by means of Theorem~\ref{theoremCoverageProbability} and Theorem~\ref{PropDetSINRDistances2}\footnote{It is worthwhile to mention that both theorems are obtained based on a single-slope path loss model but they could also be easily generalized to describe more general path-loss models such as the multi-slope path loss model used in the seminal work regarding CF massive MIMO systems~\cite{Ngo2017}. In fact, this is the topic of ongoing work by the authors, i.e., the analysis and comparison of multi-slope path loss models in CF massive MIMO systems. }. 
	\begin{table*}
		\caption{Parameters Values for Numerical Results }
		\begin{center}
			\begin{tabulary}{\columnwidth}{ | c | c | }\hline
				{\bf Description} &{\bf Values}\\ \hline
				Number users& $K=10$\\ \hline
				Number of Antennas/AP & $N=5$\\ \hline
				AP density & $\lambda_{\mathrm{AP}}=40~\mathrm{APs/km^{2}}$\\ \hline
				Communication bandwidth, carrier frequency & $W_{\mathrm{c}} = 20~\mathrm{MHz}$, $f_{0} = 1.9~\mathrm{GHz}$\\ \hline
				Uplink training transmit power per pilot symbol & ${\rho}_{\mathrm{tr}}=100~\mathrm{mW}$\\ \hline
				Downlink transmit power & ${{\rho}_{\mathrm{d}}}=200~\mathrm{mW}$ \\ \hline
				Path loss exponent & $\al=3.5$ \\ \hline
				Coherence bandwidth and time & $B_{\mathrm{c}}=200~\mathrm{KHz}$ and $T_{\mathrm{c}}=1~\mathrm{ms}$
				\\ \hline
				Duration of uplink training & $\tau_{\mathrm{tr}}=10$ samples
				\\ \hline
				Duration of uplink and downlink training is SCs & $\tau_{\mathrm{tr}}=\tau_{\mathrm{d}}=10$ samples
				\\ \hline
				Boltzmann constant & $\kappa_{\mathrm{B}}=1.381\times 10^{-23}~\mathrm{J/K}$ \\ \hline
				Noise temperature & $T_{0}= 290~\mathrm{K}$ \\ \hline
				Noise figure & ${N_\mathrm{F}}=9~\mathrm{dB}$ \\ \hline
			\end{tabulary}\label{ParameterValues1} 
		\end{center}
	\end{table*}
	
	For the sake of comparison, we consider the system model in~\cite{Papazafeiropoulos2017}, where independent users are associated with their nearest multi-antenna AP, while the remaining APs act as interferers. Henceforth, we refer to this scenario as ``small cells'' or ``SCs''. Especially, we assume that the base stations in that model have the same number of antennas serving a single user, i.e., in~\cite{Papazafeiropoulos2017}, we set $M=4$ and $K=1$ while the imperfect CSIT model in that scenario is replaced by the current one. In addition, we assume no hardware impairments and channel aging. In addition, similar to~\cite{Ngo2017}, we assume that handovers among the APs do not take place. 
	
One main difference with CF massive MIMO is that, in SCs, the effective channel power does not harden while in the case of CF massive MIMO systems, the signal power tends to its mean as the number of APs becomes large~\cite{Ngo2017}. In other words, SCs need to estimate their effective channel gain. Hence, SCs require both uplink and downlink training phases while CF massive MIMO systems rely only on uplink training. During the investigation of their performance, this difference will be more obvious. Moreover, an additional advantage, met in CF massive MIMO, is favorable propagation which can achieve optimal performance with simple linear processing. For example, on the uplink, the noise and interference can be almost canceled out with a simple linear detector such as the matched filter. Another primary reason justifying the outperformance of CF massive MIMO systems against SCs is that the latter have inherent the inter-cell interference while CF systems implement co-processing and all the APs that affect a specific user take into account for its interference. As a result, the CF approach achieves to suppress inter-cell interference by eliminating any cell boundaries~\cite{Interdonato2019}.
	
	A set of Monte Carlo simulations verifies the analytical expressions.
	In fact, by plotting the proposed analytical expressions along with the simulated results represented by means of black bullets, we observe their coincidence\footnote{In particular, the fact that the analytical results, obtained by means of the DE analysis, coincide with the simulations means that the former can be used as tight approximations in the case of a CF massive MIMO system. Although this is a known result in the massive MIMO literature~\cite{Couillet2011,Wagner2012}, the DE analysis has not been verified before as an RMT tool for CF massive MIMO systems.}. Especially, the simulated results are generated by means of the corresponding statistical SINR given by~\eqref{SINR} by averaging over $ 10^{4} $ random instances of the channels while the coverage probability and achievable rate are obtained as an average of $ 10^{4} $ realizations of different random AP topologies. The results corresponding to CF massive MIMO systems and SCs are depicted by means of ``solid'' blue and ``dot'' red lines, respectively.
	\begin{figure}[!h]
		\begin{center}
			\includegraphics[width=0.95\linewidth]{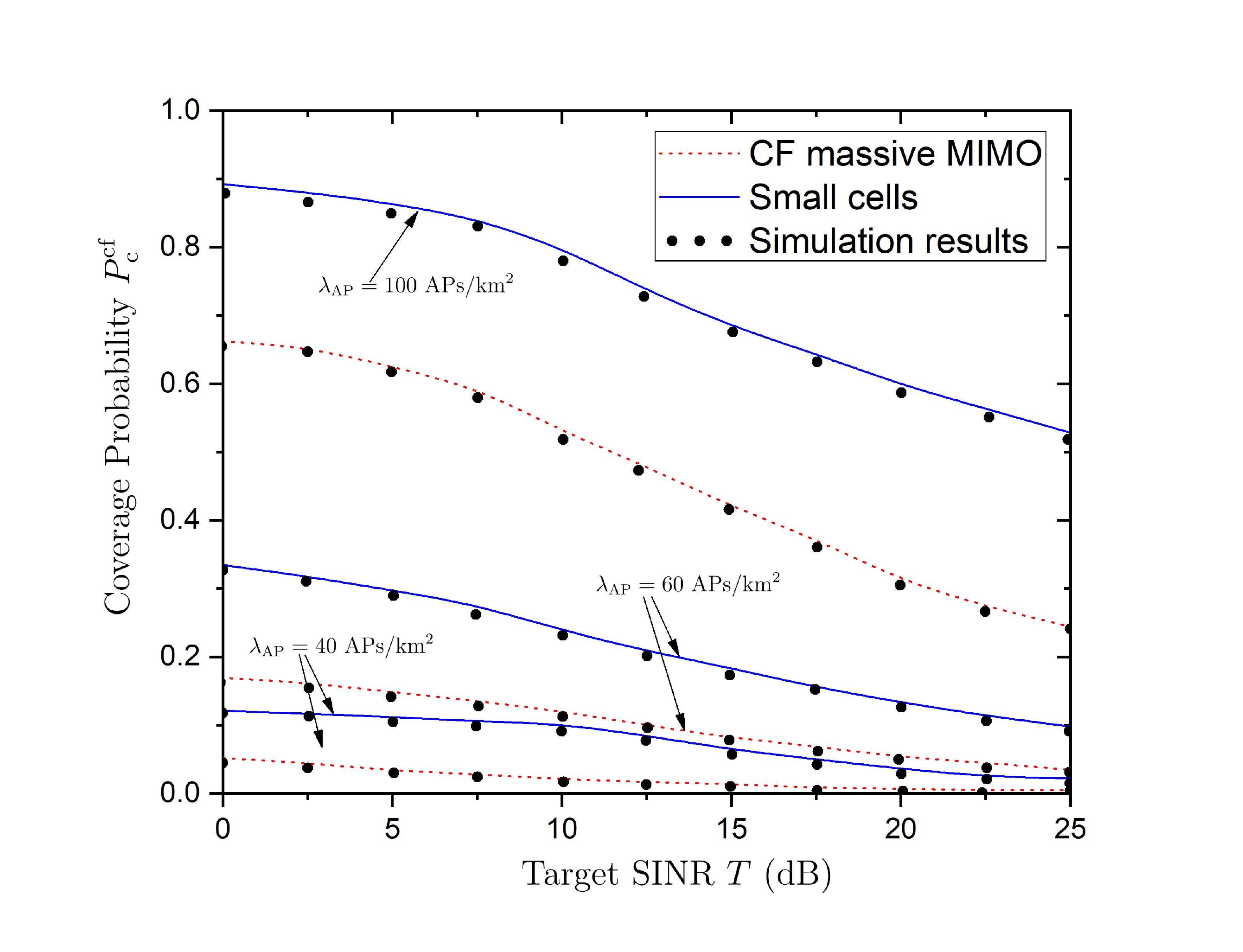}
			\caption{\footnotesize{Coverage probability for varying AP density $\lambda_{\mathrm{AP}}$ versus the target SINR $T$ for both CF massive MIMO systems and SCs.}}
			\label{Fig1}
		\end{center}
	\end{figure}
	\subsection{Setup}
	We choose a finite window of area of $1~\mathrm{km}\times 1$ $\mathrm{km}$, where we distribute the APs, each having $N=5$ antennas, according to a PPP realization with density $\lambda_{\mathrm{AP}}=40~\mathrm{APs/km^{2}}$ unless otherwise stated. Given that the analytical expressions rely on the assumption of an infinite plane while the simulation considers a finite square, we assume that this area is wrapped around at the edges to prevent any boundary effects. In addition, the structure of the system includes a number of APs serving similtaneously $K=10$ randomly distributed users. Actually, similar to~\cite{Ngo2017}, we use the default values in Table~\ref{ParameterValues1} unless otherwise stated. The normalized uplink training transmit power per pilot symbol ${\rho}_{\mathrm{tr}}$ and downlink transmit power ${{\rho}_{\mathrm{d}}}$ result by dividing $\bar{\rho}^{\mathrm{tr}}$ and ${\bar{p}_{\mathrm{d}}}$ by the noise power ${N_\mathrm{P}}$ given in $\mathrm{W}$ by ${N_\mathrm{P}} = W_{\mathrm{c}} $ $\mbox{\footnotesize{$\times$}}$ $ \kappa_{\mathrm{B}}$ $\mbox{\footnotesize{$\times$}}$ $ T_{0}$ $\mbox{\footnotesize{$\times$}}$ ${N_\mathrm{F}}$, where the various parameters are found in Table~\ref{ParameterValues1}. Also, in order to guarantee a fair comparison between CF massive MIMO systems and SCs, the total radiated power must be equal in both architectures. Hence, we have that $\bar{\rho}_{\mathrm{tr}}^{\mathrm{sc}}= \bar{\rho}_{\mathrm{tr}}$ and ${\bar{p}_{\mathrm{d}}}^{\mathrm{sc}}=\frac{M}{K}{\bar{p}_{\mathrm{d}}}$, where $\bar{\rho}_{\mathrm{tr}}^{\mathrm{sc}}$ and ${\bar{p}^{\mathrm{d}}}_{\mathrm{sc}}$ are the normalized uplink training and downlink transmit powers~\cite{Ngo2017}.

	\subsection{Depictions and Discussions}
	\subsubsection{Coverage Probability}
	The coverage probability, describing the SCs setting, is denoted by $ P_{\mathrm{c}}^{\mathrm{sc}}$ and provided by~\cite[Th.~$1$]{Papazafeiropoulos2017}.
	
	In Fig.~\ref{Fig1}, we assess the performance of the proposed bound by varying the target SINR. Specifically, firstly, it is shown the tightness of the proposed bound against the SINR. It is evident that the tightness is very good, however, it is relaxed as $ \lambda_{\mathrm{AP}}$ increases. Although someone would expect that the bound would become tighter with $ \tilde{M}\sim \lambda_{\mathrm{AP}}$ due to the use of the DE analysis, this contradiction appears due to the Alzer's inequality. Next, Fig.~\ref{Fig1} depicts that the coverage probability decreases with the target SINR in both cases of CF massive MIMO and SCs because of the inter-user and inter-cell interferences, respectively. Notably, the estimation error has its own contribution. In other words, these reasons, degrading the SINR, result in less coverage as the threshold increases. Especially, when the target SINR tends to zero, the coverage probability becomes one, when $ T \to \infty $, the coverage probability approaches zero, while, in practice, for typical values of $ T $ being around $ 15~\mathrm{dB}$, $ P_{\mathrm{c}}^{\mathrm{sc}}$ is finite and decreases. It is obvious that CF massive MIMO systems, unlike SCs, systematically provide higher coverage for all values of the target SINR $ T $ because they take benefit from favorable propagation, channel hardening, and suppression of the inter-cell interference. Actually, as the AP density $\lambda_{\mathrm{AP}}$ increases, these effects contribute more to the outperformance of CF massive MIMO systems against SCs having a cellular nature. 

	\begin{figure}[!h]
		\begin{center}
			\includegraphics[width=0.95\linewidth]{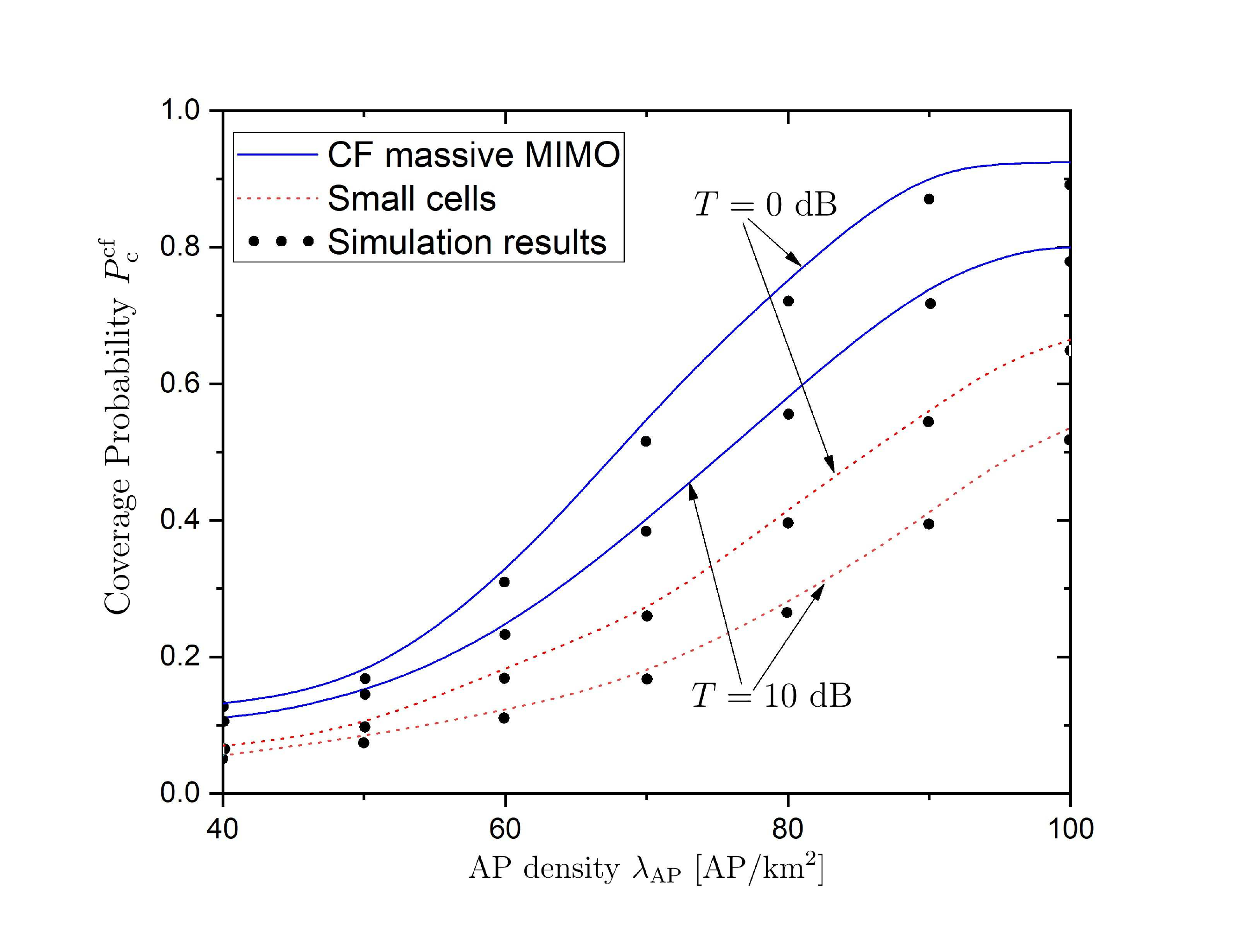}
			\caption{\footnotesize{Coverage probability for varying target SINR $T$ versus the AP density $\lambda_{\mathrm{AP}}$ for both CF massive MIMO systems and SCs.}}
			\label{Fig2}
		\end{center}
	\end{figure}
	
	\begin{figure}[!h]
		\begin{center}
			\includegraphics[width=0.95\linewidth]{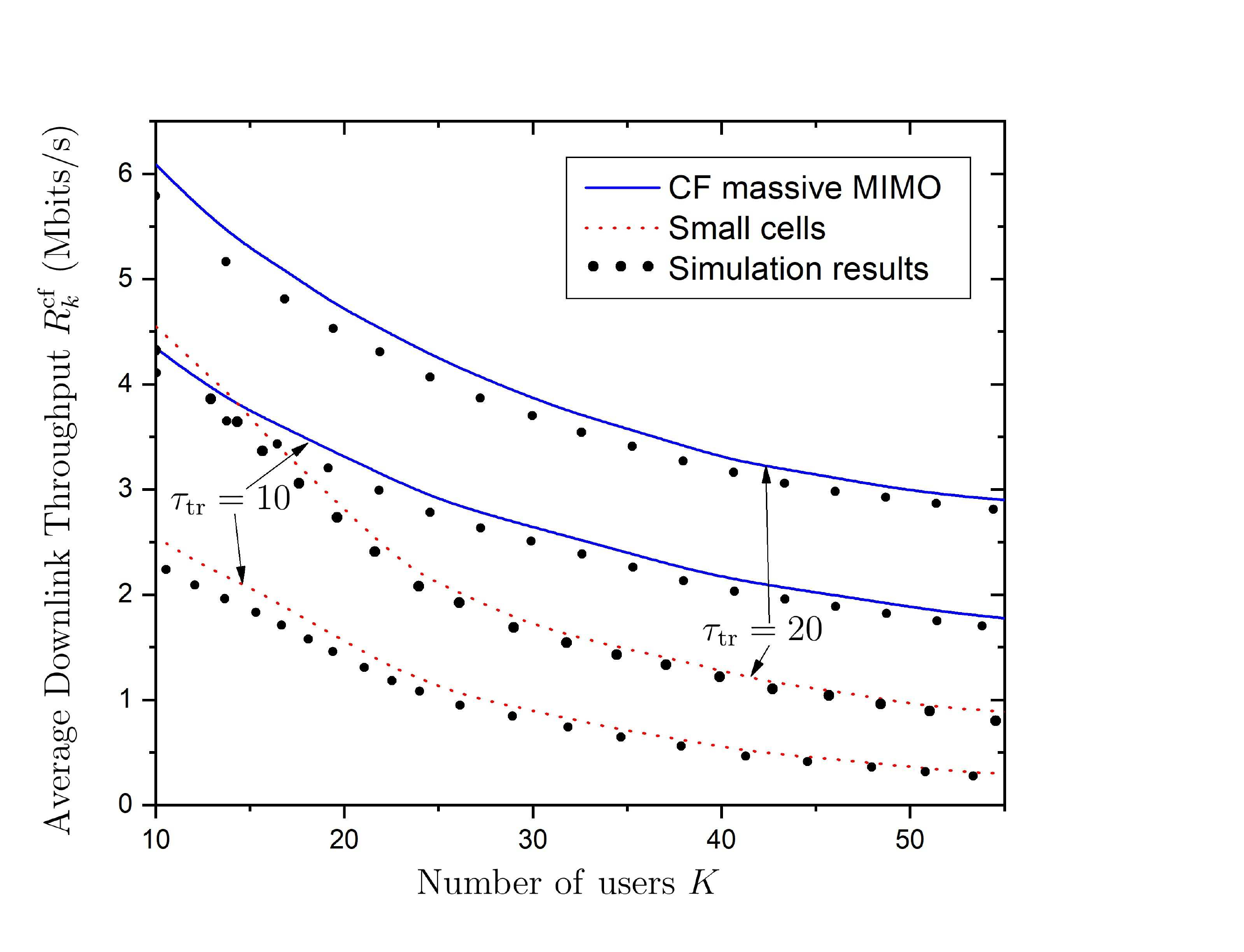}
			\caption{\footnotesize{Average downlink achievable rate for varying length of uplink training period $\tau_{\mathrm{tr}}$ versus the number of users $K$ for both CF massive MIMO systems and SCs ($\lambda_{\mathrm{AP}}=80~\mathrm{APs/km^{2}}$).}}
			\label{Fig3}
		\end{center}
	\end{figure}
	
	\begin{figure}[!h]
		\begin{center}
			\includegraphics[width=\linewidth]{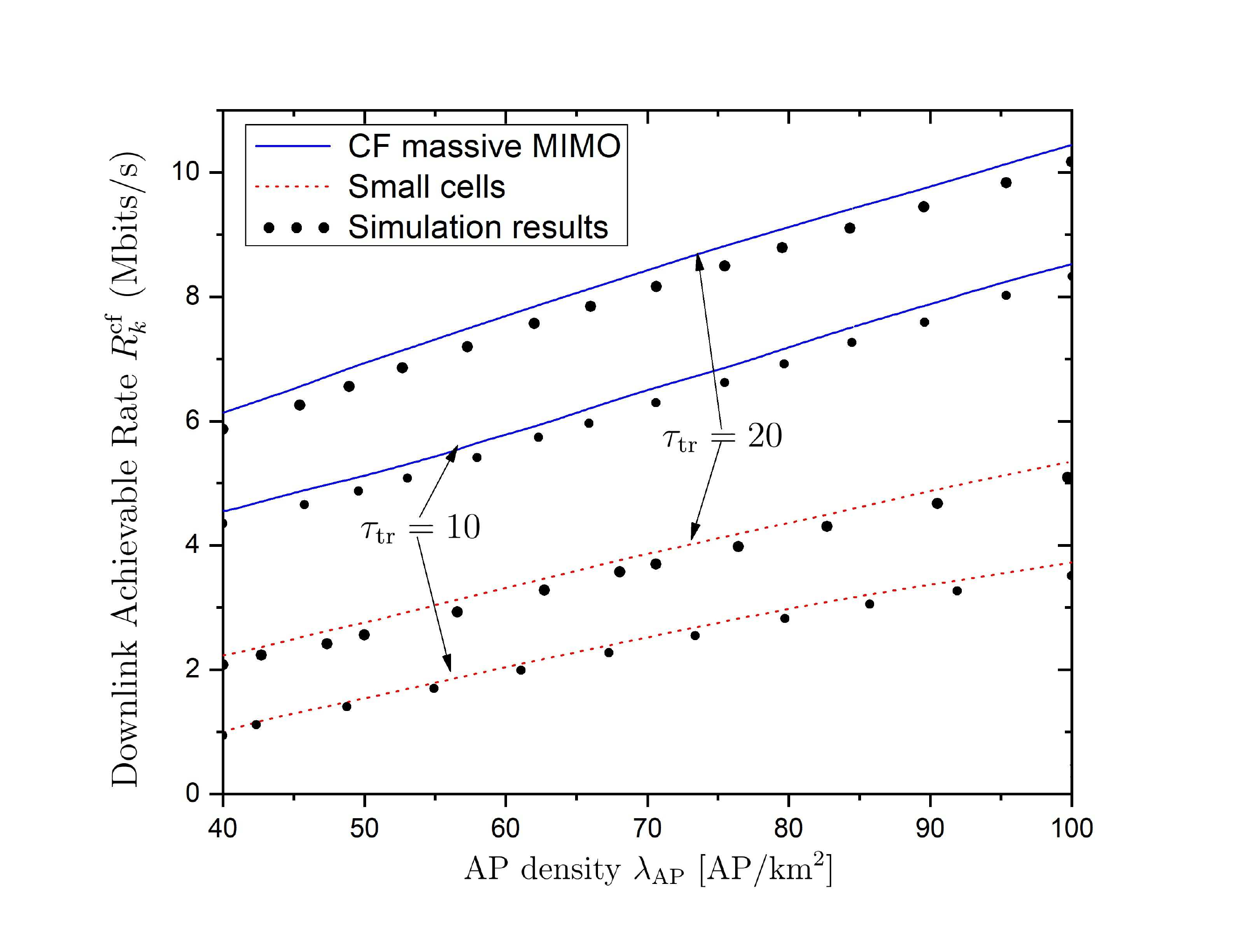}
			\caption{\footnotesize{Average downlink achievable rate for varying length of uplink training period $\tau_{\mathrm{tr}}$ versus the AP density $\lambda_{\mathrm{AP}}$ for both CF massive MIMO systems and SCs.}}
			\label{Fig4}
		\end{center}
	\end{figure}


		In Fig.~\ref{Fig2}, we investigate the impact of AP density on the coverage probability for different values of the threshold $T$. In other words, this figure allows the comparison between CF massive MIMO systems and SCs with respect to the central fundamental characteristic between the two architectures, which is the identical spatial distribution of the nodes in terms of their density. To this end, it turns out that by increasing the node density in CF massive MIMO systems and SCs, the coverage probability increases and saturates at high AP density.
		This behavior is already known for SCs in the case of single-slope path loss models~\cite{Andrews2016}, but this figure also shows the performance of CF massive MIMO systems independently, and in parallel, allows the comparison between the two network architectures. Notably, in such cases, the provided coverage by CF massive MIMO systems is higher than SCs as the density of the nodes increases regardless of the exact values of the SINR threshold because of the conditions of favorable propagation and channel hardening met in the former architecture. Moreover, a higher threshold reduces the coverage probability since it is less possible to achieve certain coverage at higher values. Furthermore, the higher the AP density, the higher the performance gap between the two architectures because CF massive MIMO systems take more advantage of cooperation among the APs and the massive MIMO property in terms of channel hardening and favorable propagation. Regarding the saturation at high AP density, this independence from $\lambda_{\mathrm{AP}}$ is the result of the SINR invariance described in~\cite{Andrews2016} and Sec.~\ref{downlink} of this work for SCs and CF massive MIMO systems, respectively. 
	
	
In Fig.~\ref{Fig3}, we study the impact of the duration of the training phase and the number of users on the achievable rate on both CF massive MIMO systems and SCs when $\lambda_{\mathrm{AP}}=80~\mathrm{APs/km^{2}}$. As expected, as the number of users $K$ increases the system performance worsens. The main source of this deterioration comes from the fact that pilot contamination becomes more severe as can be noticed by~\ref{SINR}. The same result takes place by reducing the duration of the training period $\tau_{\mathrm{tr}}$. Another main reason for the rate decrease is the multi-user interference shown in the denominator of~\ref{SINR}. Actually, the interference in SCs is more prominent because CF massive MIMO systems take advantage of the favorable propagation. Relied on this property, we observe that for a given training period the gap between CF and SC systems increases with $K$, since the interference increases. In addition, by increasing the interference, i.e., when $ K $ grows, CF massive MIMO systems perform better than SCs because the former enjoys cooperative multipoint joint processing which is more robust at higher interference. Hence, in the case that $\tau_{\mathrm{tr}}=20$ samples, the gap between CF and SCs increases from $ 1.6~\mathrm{Mbits/s} $ to $ 2.1~\mathrm{Mbits/s} $ when $ K=10 $ and $ K=55 $, respectively.
	
Fig.~\ref{Fig4} shows the achievable rate against the AP density in the cases of both CF massive MIMO and SC systems. By increasing $ \tau_{\mathrm{tr}} $, the estimated channel is improved in all cases due to less pilot contamination, and thus, the rate increases. Moreover, as anticipated, an increase in $\lambda_{\mathrm{AP}}$ increases the rate as also described in Sec.~\ref{AchievableSpectralEfficiency}, which agrees with the behavior of single-slope path loss models in SCs~\cite{Andrews2016}. Actually, the rate in both CF massive systems and SCs increases with increasing the mean number of APs due to the array gain and diversity gain, respectively, as mentioned in~\cite{Ngo2017}. However, CF massive MIMO systems present a higher rate for several reasons. In particular, CF massive MIMO systems perform much better than SCs with increasing $ \lambda_{\mathrm{AP}}$ because they take advantage of the achievable favorable propagation and channel hardening. Furthermore, as the AP density increases, the rate of CF systems is higher because the benefit from the cooperation among the APs increases. Nevertheless, the gap between the CF lines increases since the advantage from the AP cooperation increases by exploiting better the interference corresponding to a certain duration of the training phase. This property is basically justified by the reduction of the impact of pilot contamination as $\tau_{\mathrm{tr}}$ increases. In other words, CF massive MIMO systems are more robust against pilot contamination as the mean number of APs increases. Hence, when $ \lambda_{\mathrm{AP}}=40~\mathrm{APs/km^{2}}$, the gap is almost $ 1.6~\mathrm{Mbits/s} $ while when $ \lambda_{\mathrm{AP}}=100~\mathrm{APs/km^{2}}$, the gap has increased to almost $ 2~\mathrm{Mbits/s} $. At these differences of AP density, the gap is not such big but it becomes bigger when more APs are employed.
	
	\begin{figure}
		\begin{center}
			\subfigure[]{
				\includegraphics[scale=.35]{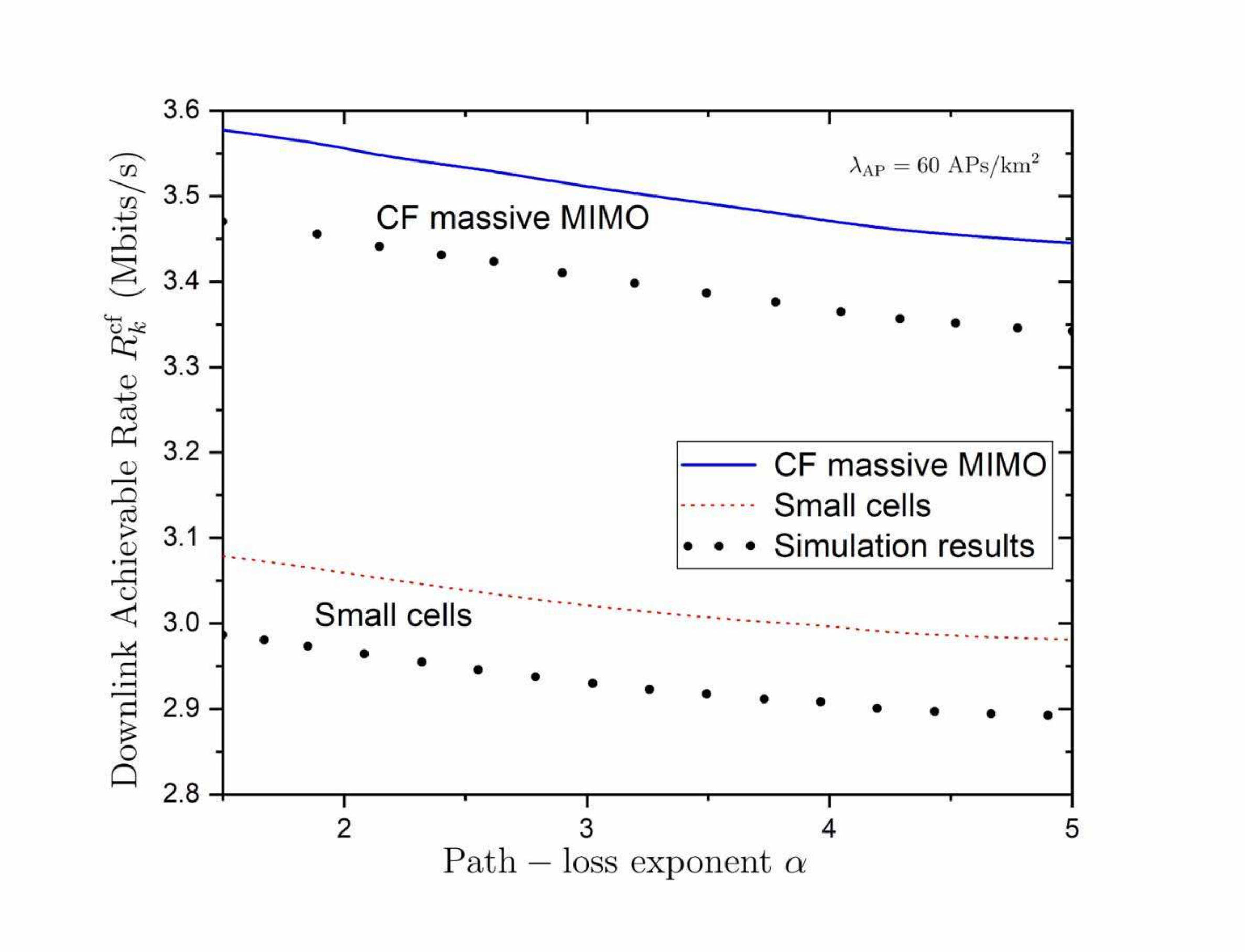}
				\label{fig_6and7combined1}
			}
			\subfigure[]{
				\includegraphics[scale=.35]{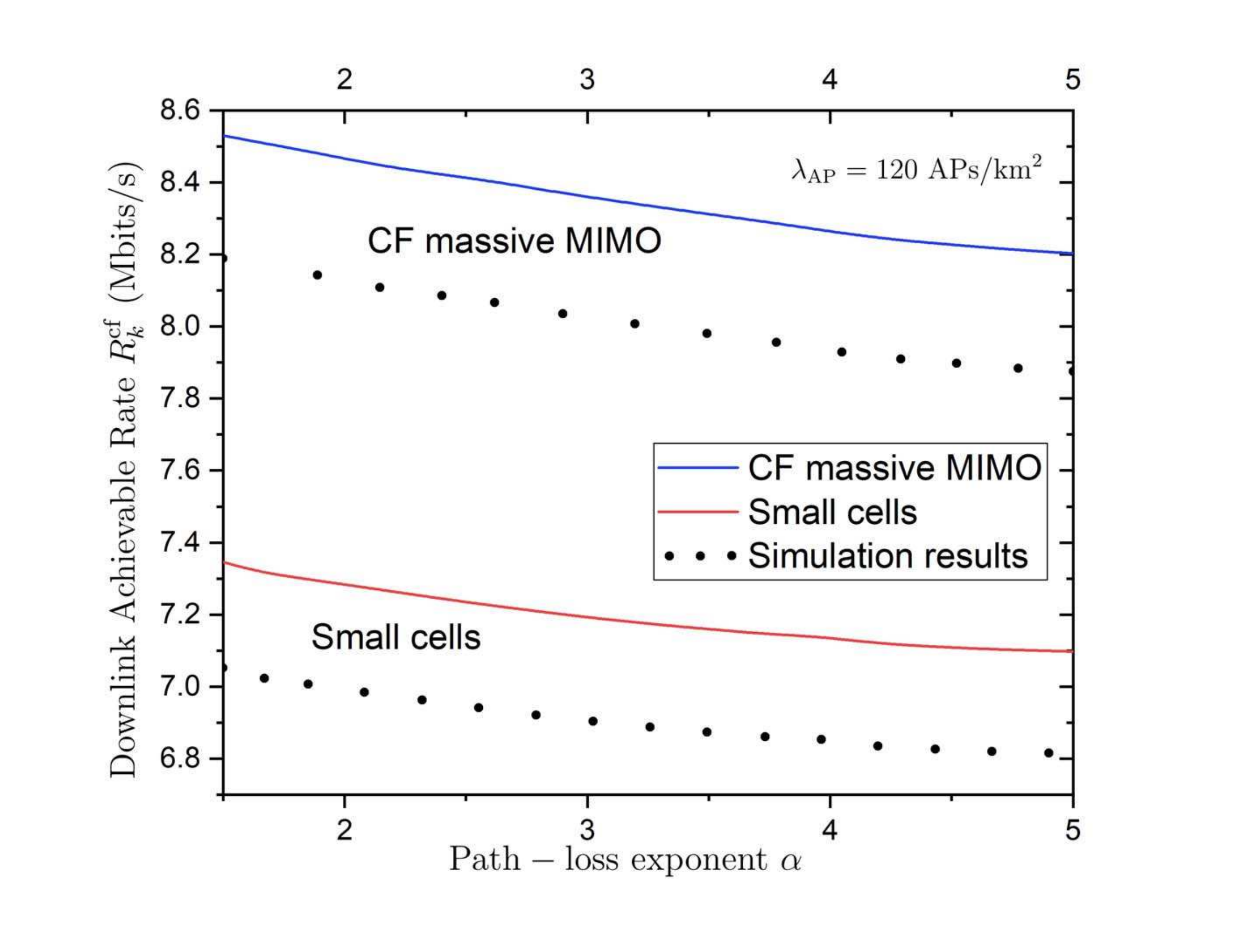}
				\label{fig_6and7combined2}
			}
		\end{center}\vspace{-15 pt}
		\caption{\footnotesize{Average downlink achievable rate for varying AP density $\lambda_{\mathrm{AP}}$ versus the path-loss exponent $\al$ for both CF massive MIMO systems and SCs. (a) $\lambda_{\mathrm{AP}}=60~\mathrm{APs/km^{2}}$, (b) $\lambda_{\mathrm{AP}}=120~\mathrm{APs/km^{2}}$.}}\vspace{-15 pt}
	\end{figure}
	
Figures \ref{fig_6and7combined1} and \ref{fig_6and7combined2} present the achievable rate versus the path-loss exponent $\al$ in the architectures of CF massive MIMO and SC systems for $\lambda_{\mathrm{AP}}=60~\mathrm{APs/km^{2}}$, and $\lambda_{\mathrm{AP}}=120~\mathrm{APs/km^{2}}$, respectively. It can be observed that the rate decreases monotonically with an increase of $\al$ for both CF massive MIMO and SC systems. Especially, the reduction of the rate is lower for larger values of $\al$, while it is higher for smaller values of $\al$. This observation implies that users far from the APs can barely affect the rate, while the users, being closer to the APs, affect strongly the rate. In parallel, these figures reveal that a larger number of APs brings an improvement in the performance of the system as explained before.

	\section{Conclusion} \label{Conclusion} 
	CF massive MIMO systems is a promising deployment paradigm for next-generation networks by embodying the distributed MIMO and massive MIMO architectures while no cell boundaries exist. In this work, given that CF massive MIMO systems have attracted a lot of attention but previous works did not account for a realistic model for the spatial randomness of the APs in the analysis despite their high irregularity, we took advantage of PPP modeling and derived tractable and closed-form expressions for the coverage probability and the achievable rate. Especially, this is the unique work providing the coverage probability of CF massive MIMO systems with PPP distributed APs.
	
	The analysis and numerical results revealed that CF massive MIMO systems outmatch SCs design with regard to both coverage and rate since it takes advantage of benefits from network MIMO and canonical massive MIMO systems. Especially, the larger the average number of APs, the higher the resultant coverage and achievable rate. Moreover, by increasing the AP density, the coverage increases up to a certain point while increasing the number of users the performance. Notably, this deterioration is less in CF massive MIMO systems exploiting the benefits of favorable propagation. Finally, the users located closer to the APs have a greater impact on the rate, and the larger the average number of APs is involved, the larger the impact eventuates.

	\begin{appendices}
		\section{Proof of Proposition~\ref{PropDetSINR}}\label{SINRproof}
		We divide each term of~\eqref{SINR} by the number $\mathcal{W}$ raised to $2$, in order to derive the correspondsing DEs. Starting with the desired signal power, we have
		\begin{align}
		S_{k}
		&= \frac{
			\mu}{\mathcal{W}^{2}} \Big|\EE\left[{\bh}_{k}^{\H}\bC_{k}\hat{\bh}_{k} \right]\!\!\Big|^{2}.\label{normalizedMRT}
		\end{align}
		First, the normalization parameter can be written by means of~\eqref{eq:lamda} and the expression of MRT precoding as
		\begin{align}
		\mu
		&= \frac{1}{\frac{1}{\mathcal{{W}}}\EE\Big[\sum_{i=1}^{K}\hat{\bh}_{i}^{\H}\bC_{i}^{2}\hat{\bh}_{i}\Big]}\nn\\
		&\asymp \left(\frac{1}{\mathcal{{W}}} \sum_{i=1}^{K}\tr \bC_{i}^{2}\bPhi_{i} \right)^{-1}\nn\\
		&= \left(\frac{1}{\mathcal{{W}}}\sum_{i=1}^{K}\tr\bC_{i} \right)^{-1}\nn\\
		&= \bar{\mu},\label{desired1MRT}
		\end{align}
		where we have applied~\cite[Thm. 3.7]{Couillet2011}\footnote{Given two infinite sequences $a_n$ and $b_n$, the relation $a_n\asymp b_n$ is equivalent to $a_n - b_n \xrightarrow[ n \rightarrow \infty]{\mbox{a.s.}} 0$.}.
		Note that $\bH\!=\! \big[\bh_{1},\ldots, \bh_{K} \big] \!\in \bbC^{\mathcal{W} \times K}$ is the channel matrix from the APs to all users. The DE of~\eqref{normalizedMRT} is obtained as
		\begin{align}
		\frac{1}{\mathcal{W}}\EE\!\left[{\bh}_{k}^{\H}\bC_{k}\hat{\bh}_{k} \right]
		&\!=\! \frac{1}{\mathcal{W}}\EE\!\left[\left(\hat{\bh}^\H_{k} \!+ \!\tilde{\bee}^\H_{k} \right)\bC_{k}\hat{\bh}_{k} \right]\label{desired2MRT}\\
		&=\frac{1}{\mathcal{W}}\EE\Big[{\hat{\bh}^\H_{k}\bC_{k}\hat{\bh}_{k} }\Big]\nn\\
		&\asymp\frac{1}{\mathcal{W}}\tr \bC_{k}\bPhi_{k}\nn\\
		&=1,\label{desired4MRT}
		\end{align}
		where in~\eqref{desired2MRT} we have taken into account that $\hat{\bg}_{k}$ and $\tilde{\bee}_{k}$ are uncorrelated, and next we have applied~\cite[Thm. 3.7]{Couillet2011} since all conditions are satisfied. Note that the matrices commute because they are diagonal. Therefore, the DE signal power $\bar{S}_{k} = \lim_{\mathcal{W} \rightarrow \infty} S_{k}$ is written as
		\begin{align}
		\bar{S}_{k} =\bar{\mu}.\label{eq:theorem4.5MRT}
		\end{align}
		This result verifies the chosen scaling regarding the precoder.
		Next, we focus on the derivation of DEs of the denominator terms. The first term, involving the variance, is obtained as
		\begin{align}
		\frac{1}{\mathcal{W}^{2}} \var\left[ {\bh}_{k}^{\H}\bC_{k}\hat{\bh}_{k}\right]-\frac{1}{\mathcal{W}^{2}} \EE\bigg[\Big|{\tilde{\bee}_{k}^{\H}\bC_{k}\hatvh
			_{k} }\Big|^2\bigg]\xrightarrow[ \mathcal{W} \rightarrow \infty]{\mbox{a.s.}} 0.\label{DEvariance} 
		\end{align}
		In~\eqref{DEvariance}, we have exploited the property of the variance operator $\mathrm{var}\left[ x \right]=\EE[x^{2}]- \EE^{2}[x] $ and that $\tilde{\bee}_{k}={\bh}_{k}-\hat{\bh}_{k}$. In addition, we have applied~\cite[Thm. 3.7]{Couillet2011}. After applying again this theorem, we have
		\begin{align}
		\frac{1}{\mathcal{W}^{2}} \EE\bigg[\Big|{\tilde{\bee}_{k}^{\H}\bC_{k}\hatvh
			_{k} }\Big|^2\bigg]
		&\asymp \frac{{1}}{\mathcal{W}^{2}} \tr \bC^{2}_{k} \bPhi_{k}\left( \bL_{k}- \bPhi_{k}\right) \nn\\
		&=\frac{{1}}{\mathcal{W}^{2}} \tr \left( \bD\bL_{k}^{-1}- \Id_{\mathcal{W}}\right). \label{eq:theorem4.7MRT}
		\end{align}
		The final term becomes 
		\begin{align}
		\frac{1}{\mathcal{W}^{2}}\EE\left[ \Big|{\bh}_{k}^{\H}\bC_{i}\hat{\bh}_{i}\Big|^{2}\right] 
		&\asymp\frac{1}{\mathcal{W}^{2}}\tr\bC_{i}^{2}{\bPhi}_{i}\bL_{k} \nn\\
		&=\frac{1}{\mathcal{W}^{2}}\tr \bD\bL_{i}^{-2}\bL_{k}\label{lasttermMRT}
		\end{align}
		since ${\bh}_{k} $ and $\hatvh_{i} $ are mutually independent. 
		Taking into account that the SINR is conditioned on $\bL_{k}$, substitution of~\eqref{desired1MRT},~\eqref{eq:theorem4.5MRT},~\eqref{eq:theorem4.7MRT}, and~\eqref{lasttermMRT} into~\eqref{SINR} completes the proof.
		
		\section{Proof of Theorem~\ref{theoremCoverageProbability}}\label{CoverageProbabilityproof}
		The proof starts by writing the terms of~\eqref{DESINR}, including the block matrix traces, as summations over the diagonal elements (element-wise). Thus, the DE SINR, conditioned on the distances $r_{mi}$ for $i=1,\ldots, K$, is obtained as
		\begin{align}
		\bar{\gamma}_{k}\asymp\frac{ M N}{ \frac{{1}}{M} \sum_{i=1}^{K}\sum_{m=1}^{M} d_{m}l_{mi}^{-2}\left( l_{mk}+\frac{MN}{\rho_{\mathrm{d}}}\right)-1
		}.\label{gamma1} 
		\end{align}
		We continue with the derivation of distribution of the SINR, conditioned on a realization of $l_{mi}$ for $i=1,\ldots, K$, i.e., ${\mathbb{P}}\!\left( \bar{\gamma}_{k}>T|l_{m1},\ldots,l_{mK}\right)$. Specifically, after substituting~\eqref{gamma1} inside the expression of the coverage probability, and by means of several algebraic manipulations, we obtain~\eqref{coverage5}. Hence, the conditional coverage probability is written as shown at the top of next page
		\begin{longequation*}[tp]

			\begin{align}
			& \!\!\!{\mathbb{P}}\left(\bar{\gamma}_{k}>\!T|r_{m1},\ldots,r_{mK}\right) \!=
			{\mathbb{P}}\Bigg( \!\!\mathcal{W}\!>{T}\Bigg(\!\! \frac{1}{M}\sum_{i=1}^{K}\sum_{m=1}^{M} d_{m}l_{mi}^{-2}\left( l_{mk}+\frac{MN}{ \rho_{\mathrm{d}}}\right)-1\!\!\Bigg)\!\!\Bigg) \label{coverage5}\\
			& \!\!\!\approx 
			\tilde{\mathbb{P}}\Bigg( \tilde{g}\!>{T}\Bigg( \frac{1}{M} \sum_{i=1}^{K}\sum_{m=1}^{M} d_{m}l_{mi}^{-2}\left( l_{mk}+\frac{MN}{ \rho_{\mathrm{d}}}\right)-1\Bigg)\!\!\Bigg)\label{coverage6}\\
			&\!\!\!\approx 1-\!\Bigg(\!1-\exp\Bigg( -\eta {T} \Bigg( \frac{1}{M} \sum_{i=1}^{K}\sum_{m=1}^{M} d_{m}l_{mi}^{-2}\left( l_{mk}+\frac{MN}{ \rho_{\mathrm{d}}}\right)-1\Bigg)\!\! \Bigg)\!\! \Bigg)^{\!\tilde{\mathcal{W}}} \label{coverage71}\\
			&\!\!\!= \sum^{\tilde{\mathcal{W}}}_{n=1} \!\binom{\tilde{\mathcal{W}}}{n}\!\left( -1 \right)^{n+1} \exp\!\bigg(\!\! \!-n \eta {T}\Bigg(\!\! \frac{1}{M} \sum_{i=1}^{K}\sum_{m=1}^{M} d_{m}l_{mi}^{-2}\left( l_{mk}+\frac{MN}{ \rho_{\mathrm{d}}}\right)-1\Bigg)\!\! \Bigg)\!.\!\! \label{coverage7}
			\end{align} 
			\hrule
		\end{longequation*}
		
		In~\eqref{coverage6}, we have approximated the constant number $\mathcal{W}$ by considering the dummy gamma variable $\tilde{g}$, having mean $\mathcal{W}=MN$ and shape parameter $\tilde{\mathcal{W}}=\EE\left[ W\right] = \tilde{M} N$. This approximation becomes tighter as $\tilde{\mathcal{W}}$ goes to infinity~\cite{Alzer1997}, since $\lim_{y \to \infty}\frac{y^{y}x^{y-1}\mathrm{e}^{-yx}}{\Gamma\left( y \right)}=\delta\left( x-1 \right)$ with $\delta\left( x \right)$ being Dirac's delta function. Notably, this approximation, used in~\cite{Bai2016}, becomes more precise in our system model involving a large number (massive) of APs. Note that the precision increases as the number of antennas per AP increases. In~\eqref{coverage71}, we have applied Alzer's inequality (see~\cite[Lemma~1]{Alzer1997}), where $\eta=\tilde{\mathcal{W}} \left( \tilde{\mathcal{W}}! \right)^{-\frac{1}{\tilde{\mathcal{W}}}}$, while afterwards, we have used the Binomial theorem. Note that~\eqref{coverage71} does not contain any random variable since this expression is conditioned on the distances. 
		Next, the coverage probability is obtained by evaluating the expectation of~\eqref{coverage7} with respect to AP locations given that the distances between the APs and the users are uniformly distributed. Thus, we have
		\begin{align}
		P_{\mathrm{c}}^{\mathrm{cf}}\!&=\!\sum^{\tilde{\mathcal{W}}}_{n=1} \!\binom{\tilde{\mathcal{W}}}{n}\!\left( -1 \right)^{n+1}\nn\\
		&\times \EE\left[ \exp\!\bigg(\!\! -n \eta {T}\Bigg( \frac{1}{M} \sum_{i=1}^{K}\sum_{m=1}^{M} \mathcal{I}_{mk}-1\Bigg)\!\! \Bigg)\!\! \right]\label{pc2}\\
		&\!\ge\! 
		\sum^{\tilde{\mathcal{W}}}_{n=1} \!\binom{\tilde{\mathcal{W}}}{n}\!\left( -1 \right)^{n+1} e^{ {n\eta {T}\lambda_{\mathrm{AP}}}{ } } \nn\\
		&\times\exp\!\bigg(\!\! -{n\eta 
			{T}}\,\EE\left[\frac{1}{M} \sum_{i=1}^{K}\sum_{m=1}^{M} \mathcal{I}_{mk}\right] \Bigg),\label{pcproof} 
		\end{align}
		where we have set $\mathcal{I}_{mk}= d_{m}l_{mi}^{-2}\left( l_{mk}+\frac{MN}{ \rho_{\mathrm{d}}}\right)$ and have applied Jensen's inequality since $\mathrm{exp}\left( \cdot \right)$ is a convex function. 
		By focusing on the derivation of the expectation, we have 
		\begin{align}
		&\lim_{R \to \infty} \! \EE\Bigg[\frac{1}{M} \sum_{i=1}^{K}\sum_{m=1}^{M} \mathcal{I}_{mk}\Bigg] \nn\\
		&=\lim_{R \to \infty} \EE_{M}\!\left[\!\EE_{|M} \!\!\left[\frac{1}{{M}}{\displaystyle\sum_{i=1}^{K}\!\sum_{m\in \Phi_{\mathrm{AP}}\cap B\left( o,R \right)}^{M}\! \!\!\!\!\!\!\!\!\!\mathcal{I}_{mk}}{}|M=\Phi\!\left( B\left( o,R \right) \right)\right]\right] \label{SINRFiniteproof141}\\
		&= \sum_{i=1}^{K}\lim_{R \to \infty} \EE_{M}\!\left[\mathcal{I}_{mk}\right] \label{SINRFiniteproof151}\\
		&= \sum_{i=1}^{K}\EE \left[\left( \sum_{j=1}^{K}|\bpsi_{j}^{\H}\bpsi_{k}|^{2}l_{mj}+\frac{1}{{\tau_{\mathrm{tr}} \rho_{\mathrm{tr}}}} \right)\!\!\left( l_{mk}+\frac{\EE[M]N}{ \rho_{\mathrm{d}}}\right)l_{mi}^{-2}\right] \label{SINRFiniteproof161}
		\end{align}
		where in~\eqref{SINRFiniteproof141}, we have assumed a ball of radius $R$ centered at the origin that contains $M=\Phi\left( B\!\left( o,R \right) \right)$ points with $ S\!\left( \mathcal{A} \right)=|B\!\left( o,R \right)\!|$. By conditioning on this area of radius $R$ and on the number of points in this area, $M$ in the denominator cancels out with the number of points inside the ball. In~\eqref{SINRFiniteproof161}, we have substituted $\mathcal{I}_{mk}= d_{m}l_{mi}^{-2}\left( l_{mk}+\frac{MN}{ \rho_{\mathrm{d}}}\right)$ and $d_{m}$.
		Then, we substitute $d_{m}$, and we result in
		\begin{align}
		&\!\!\!\!\!\!\mathcal{I}_{1}\!= \!\EE\left[ \sum_{i=1}^{K}\left( \sum_{j=1}^{K}|\bpsi_{j}^{\H}\bpsi_{k}|^{2}l_{mj}+\frac{1}{{\tau_{\mathrm{tr}} \rho_{\mathrm{tr}}}} \right)l_{mi}^{-2}l_{mk}\right]\nn\\
		&\!\!\!\!\!\!=\!\EE\!\!\left[ \!\sum_{i=1}^{K}\! \sum_{j=1}^{K}\!|\bpsi_{j}^{\H}\bpsi_{k}|^{2} l_{mj}l_{mi}^{-2}l_{mk}\!\right]\!+\!\frac{1}{{\tau_{\mathrm{tr}} \rho_{\mathrm{tr}}}}\EE\!\left[ \sum_{i=1}^{K} l_{mi}^{-2}l_{mk}\right]\label{Expectationequation} 
		\end{align}
		and 
		\begin{align}
		\mathcal{I}_{2}&= \frac{ \lambda_{\mathrm{AP}} N}{ \rho_{\mathrm{d}}}\EE\left[ \sum_{i=1}^{K}\left( \sum_{j=1}^{K}|\bpsi_{j}^{\H}\bpsi_{k}|^{2}l_{mj}+\frac{1}{{\tau_{\mathrm{tr}} \rho_{\mathrm{tr}}}} \right)l_{mi}^{-2}\right].\label{Expectationequation12} 
		\end{align}
		Regarding the first part of~\eqref{Expectationequation}, we have
		\begin{align}
		\!\!&\EE\left[\sum_{i=1}^{K} \sum_{j=1}^{K}|\bpsi_{j}^{\H}\bpsi_{k}|^{2}l_{mj}l_{mi}^{-2}l_{mk}\right]\nn\\
		&=\left\{\begin{array}{ll}
		\sum_{i=1}^{K} |\bpsi_{i}^{\H}\bpsi_{k}|^{2} \EE \left[ l_{mi}^{-1}l_{mk}\right]&~\mathrm{if}~j=i\\
		\sum_{i=1}^{K} \EE \left[ l_{mi}^{-2}l_{mk}^{2} \right] &~\mathrm{if}~j= k\\
		\sum_{j\ne i,k}^{K}
		|\bpsi_{j}^{\H}\bpsi_{k}|^{2}\EE \left[ l_{mj} l_{mi}^{-2}l_{mk} \right] &~\mathrm{otherwise}
		\end{array}\!.\label{Expectationequation1}
		\right. 
		\end{align}
		The expectation in the first branch of the right hand side of~\eqref{Expectationequation1} for $i\ne k$ gives
		\begin{align}
		\EE \left[ l_{mi}^{-1}l_{mk} \right]
		& \ge \frac{1}{\EE \left[ l_{mi}^{}\right]}\EE\left[ l_{mk}^{}\right]\label{firstPart0} \\
		&=1 \label{firstPart2},
		\end{align}
		where~\eqref{firstPart0} takes advantage of Jensen's inequality, and then,~\eqref{firstPart2} is obtained since the two variables have the same marginal distribution. By following similar steps, the derivation of the expectation in the second branch is straightforward, while the last branch becomes
		\begin{align}
		\EE \left[ l_{mj}^{} l_{mi}^{-2}l_{mk}^{} \right] &= \left\{\begin{array}{ll}
		\EE \left[ l_{mj}^{} l_{mk}^{-1} \right]&~\mathrm{if}~i=k\\
		\EE \left[ l_{mj}^{} l_{mi}^{-2}l_{mk}^{} \right]&~\mathrm{if}~i\ne k
		\end{array}\label{Expectationequation2}
		\right. .
		\end{align}
		If $i=k$, the expression in the first branch is identical to~\eqref{firstPart0}, and the result is the same. The remaining term in~\eqref{Expectationequation2} is written as
		\begin{align}
		\EE \left[ l_{mj}^{} l_{mi}^{-2}l_{mk}^{} \right]&=\EE \left[ l_{mj}^{} \right] \EE\left[ l_{mi}^{-2}\right] \EE\left[ l_{mk}^{} \right]\label{thirdPart1}\\
		&\ge \EE \left[ l_{mj}^{} \right] \EE\left[ l_{mi}^{-1}\right] ^{2}\EE\left[ l_{mk}^{} \right]\label{thirdPart2}\\
		&\ge 1,\label{thirdPart3}
		\end{align}
		where~\eqref{thirdPart1} considers the independence among the variables, while~\eqref{thirdPart2} exploits the inequality $\EE\left[ x^{2}\right] \ge \EE\left[ x\right]^{2} $. Last,~\eqref{thirdPart3} follows basically the same steps as those
		taken in~\eqref{firstPart2}.
		The second part of~\eqref{Expectationequation} becomes
		\begin{align}
		\EE\left[ \sum_{i=1}^{K} l_{mi}^{-2}l_{mk}\right]&= \left\{\begin{array}{ll}
		\EE\left[ l_{mi}^{-1}\right]&~\mathrm{if}~i=k \\
		\sum_{i\ne k}^{K}\EE\left[ l_{mi}^{-2}l_{mk}^{}\right]&~\mathrm{if}~i\ne k \end{array}.\label{Expectationequation3}
		\right. \end{align}
		Let us now tackle both expectations separately. The former, i.e., $\EE\left[ l_{mi}^{-v}\right]$ for $v=1$ results in
		\begin{align}
		\EE\left[ l_{mi}^{-v}\right]&=\EE\left[ \frac{1}{l_{mi}^{v}}\right]\label{secondpart} \\
		&\ge \frac{1}{\EE\left[l_{mi}^{v}\right]},
		\end{align}
		where Jensen's inequality has been applied in~\eqref{secondpart}. The final expression is obtained by computing $\EE\left[l_{mi}^{v}\right]$ as
		\begin{align}
		\EE\left[l_{mi}^{v}\right]&=2 \pi \left( \int_{0}^{1}{ y} \mathrm{d}y+\int_{1}^{ \infty} y^{-va+1} \mathrm{d}y \right) \label{eq51} \\
		&=\frac{ v \al \pi }{v \al-2 }. \label{eq6}
		\end{align}
		The latter expectation in~\eqref{Expectationequation3} is computed as
		\begin{align}
		\EE\left[ l_{mi}^{-2}l_{mk}^{}\right]&=\EE\left[ l_{mi}^{-2}\right] \EE\left[ l_{mk}^{}\right]\\
		&\ge \EE\left[ l_{mi}^{-1}\right]^{2} \EE\left[ l_{mk}^{}\right]\\
		&\ge \frac{\EE\left[ l_{mk}^{}\right]}{\EE\left[ l_{mi}^{}\right]^{2}}\\
		&=\frac{1}{\EE\left[ l_{mi}^{}\right]}\\
		&=\frac{ \al-2 }{ \al \pi},
		\end{align}
		where we have used similar techniques as before. 
		%
		By substituting all these expressions in~\eqref{Expectationequation}, we obtain $\mathcal{I}_{1}$. Similarly, $\mathcal{I}_{2}$ is obtained as
		\begin{align}
		\mathcal{I}_{2}&= \frac{K \lambda_{\mathrm{AP}} N}{ \rho_{\mathrm{d}}\al \pi}\left(
		\sum_{j=1}^{K} |\bpsi_{j}^{\H}\bpsi_{k}|^{2} \left( \al-2 \right)+\frac{\al-1}{{\tau_{\mathrm{tr}} \rho_{\mathrm{tr}}}}\right) \label{Expectationequation13}.
		\end{align}
		Having derived $\mathcal{I}_{1}$ and $\mathcal{I}_{2}$, we substitute their expressions in~\eqref{SINRFiniteproof161}, and we eventually complete the proof resulting first in~\eqref{pc1}, and next, in~\eqref{pc2} after using the binomial theorem.

		\section{Proof of Theorem~\ref{PropDetSINRDistances2}}\label{SINRproofDistances2}
		The proof is split in two subsections. In the first subsection, we provide a more tractable bound than~\eqref{Ratebar} that will allow to average over a PPP realization of the APs, while the second subsection includes the derivation of the PPP averaged inverse SINR.
		\subsection{Lower bound of the downlink SE}\label{LowerBoundDownlinkSE}
		Rewriting~\eqref{Ratebar} by means of the inverse of $\gamma_{k}$, and applying the Jensen inequality we have 
		\begin{align}
		\EE\left[\log_{2}\left( 1+\frac{1}{\gamma_{k}^{-1}} \right) \right] \ge \log_{2} \left( 1+\check{\gamma}_{k} \right),
		\end{align}
		where the expectation applies directly to the inverse of the SINR since $\check{\gamma}_{k}=\frac{1}{\EE\left[\gamma_{k}^{-1} \right]}$.
		\subsection{Derivation of $\check{\gamma}_{k}$}\label{LowerBoundDownlinkSE}
		After writing the trace of each matrix as the sum of its entry-wise elements, the expectation of the inverse of the SINR, given by~\eqref{gamma1}, is written as

		\begin{align}
		\EE\left[ \gamma_{k}^{-1} \right] &= \frac{1}{N}\EE\left[ \frac{{1}}{M^{2}} \left(\sum_{i=1}^{K} \sum_{m=1}^{M} \mathcal{I}_{mk}-M\right)\right] .\label{SINRFiniteproof11}
		\end{align}
		We are going to compute the expectation by considering a ball of radius $R$ centered at the origin including $M=\Phi\left( B\left( o,R \right) \right)$ points with $ S\!\left( \mathcal{A} \right)=|B\left( o,R \right)|$. Then, conditioning on this area of radius $R$ and on the number of points in this finite area, the application of the law of large numbers will take place. In the next step, we remove the conditioning regarding the number of points while we let $R \to \infty$, i.e., the area goes to infinity. Specifically, the expectation in the previous expression becomes
		\begin{align}
		&\EE\left[ \frac{{1}}{M^{2}} \left(\sum_{i=1}^{K} \sum_{m\in \Phi_{\mathrm{AP}}} \mathcal{I}_{mk}-M\right)\right]\label{SINRFiniteproof12} \\
		&=\lim_{R \to \infty} \EE\left[ \frac{{1}}{M^{2}} \left( \sum_{i=1}^{K} \sum_{m\in \Phi_{\mathrm{AP}}\cap B\left( o,R \right)} \mathcal{I}_{mk} -M \right)\right]\label{SINRFiniteproof13}\\
		&=\!\lim_{R \to \infty}\!\!\EE_{M}\bigg[\!\EE_{|M} \!\!\left[\! \frac{{1}}{M^{2}} \!\! \sum_{i=1}^{K} \sum_{m\in \Phi_{\mathrm{AP}}\cap B\left( o,R \right)} \!\! \!\!\!\!\!\!\!\!\!\mathcal{I}_{mk}|M=\Phi\!\left( B\left( o,R \right) \right)\right]\nn\\
		&\left. -\frac{1}{M}\right]\label{SINRFiniteproof14}\\
		&\approx \lim_{R \to \infty} \frac{1}{\EE_{M}[M]}\EE \Bigg[\sum_{i=1}^{K} l_{mi}^{-2} \left( \sum_{j=1}^{K}|\bpsi_{j}^{\H}\bpsi_{k}|^{2}l_{mj}\!+\!\frac{1}{{\tau_{\mathrm{tr}} \rho_{\mathrm{tr}}}} \right)\!\! l_{mk}\nn\\
		&+\frac{N}{\rho_\mathrm{d}} \sum_{i=1}^{K} l_{mi}^{-2}\!\! \left( \sum_{j=1}^{K}|\bpsi_{j}^{\H}\bpsi_{k}|^{2}l_{mj}\!+\!\frac{1}{{\tau_{\mathrm{tr}} \rho_{\mathrm{tr}}}} \right)\!\!\Bigg] \!-\!\frac{1}{\EE_{M}[M]},\label{SINRFiniteproof16}
		\end{align}
		where in \eqref{SINRFiniteproof13}, we have written the previous equation in terms of the ball of radius $R$. In \eqref{SINRFiniteproof14}, we condition on the number of points inside the ball. Then, given that the SINR has been derived by means of the DE analysis, which holds for $M \to \infty$, we are able to apply~\cite[Lemma 1]{Zhang2014a}. Thus, in \eqref{SINRFiniteproof16}, we have applied this lemma. Next, we have $\EE_{M}\left[ M\right]= \lambda_{\mathrm{AP}} |B\left( o,R \right)|$ while the other expectations in \eqref{SINRFiniteproof16} have already been derived in parts in Appendix~\ref{CoverageProbabilityproof}. Hence, $\check{\gamma}_{k}$ is obtained, and the proof is concluded.

	\end{appendices}

	\bibliographystyle{IEEEtran}

	\bibliography{mybib}

\end{document}